\documentclass[10pt]{amsart}
\usepackage{amsmath, amsfonts, amsthm, amssymb}
\usepackage[margin=1.5in]{geometry}
\usepackage{epstopdf}
\usepackage{graphicx}
\usepackage{url}
\usepackage{verbatim}
\usepackage{hyperref}
\usepackage{xcolor}
\usepackage{enumerate}
\graphicspath{{Figures/}}

\usepackage{tabularx,longtable,multirow}

\newtheorem{theorem}{Theorem}[section]

\newtheorem{lemma}[theorem]{Lemma}
\newtheorem{proposition}[theorem]{Proposition}

\theoremstyle{definition}
\newtheorem{definition}[theorem]{Definition}
\newtheorem{remark}[theorem]{Remark}

\numberwithin{equation}{section}

\DeclareMathOperator*{\argmax}{arg\,max}
\newcommand{\ind}{1\hspace{-2.1mm}{1}}

\newcommand{\RR}{\mathbb{R}}
\newcommand{\PP}{\mathbb{P}}

\newcommand{\E}{\mathrm{e}}
\newcommand{\I}{\mathrm{i}}
\newcommand{\cop}{\mathrm{c}}
\newcommand{\Cop}{\mathrm{C}}
\newcommand{\Op}{\mathfrak{O}}
\newcommand{\Cf}{\mathfrak{C}}
\newcommand{\Pf}{\mathfrak{P}}
\newcommand{\Cft}{\widetilde{\mathfrak{C}}}
\newcommand{\Pft}{\widetilde{\mathfrak{P}}}
\newcommand{\Sf}{\mathfrak{S}}

\newcommand{\Ll}{\mathcal{L}}

\newcommand{\EE}{\mathbb{E}}
\newcommand{\VV}{\mathbb{V}}

\newcommand{\LLa}{\boldsymbol{\Lambda}}

\newcommand{\Qm}{\boldsymbol{Q}}

\newcommand{\taub}{\boldsymbol{\tau}}
\newcommand{\alb}{\boldsymbol{\alpha}}
\newcommand{\thb}{\boldsymbol{\theta}}
\newcommand{\mum}{\boldsymbol{\mathrm{\mu}}}
\newcommand{\uu}{\boldsymbol{\mathrm{u}}}
\newcommand{\ww}{\boldsymbol{\mathrm{w}}}
\newcommand{\xx}{\boldsymbol{\mathrm{x}}}

\newcommand{\Eb}{\boldsymbol{\mathrm{E}}}

\newcommand{\Xb}{\boldsymbol{\mathrm{X}}}
\newcommand{\Fb}{\boldsymbol{\mathrm{F}}}

\newcommand{\UnitVec}{\boldsymbol{\mathrm{1}}}
\newcommand{\ZeroVec}{\boldsymbol{\mathrm{0}}}
\newcommand{\Rb}{\boldsymbol{R}}

\newcommand{\Sigmab}{\boldsymbol{\mathrm{\Sigma}}}

\newcommand{\FInv}{F^{\leftarrow}}

\graphicspath{{Figures/}}
\begin{document}

\title{Portfolio optimisation with European options}

\author{Jonathan Raimana Chan}
\address{Kaiju Capital Management}
\email{chanjonathan.689@gmail.com}

\author{Thomas Huckle}
\address{Kaiju Capital Management}
\email{tom.huckle20@imperial.ac.uk}

\author{Antoine Jacquier}
\address{Department of Mathematics, Imperial College London, and Alan Turing Institute}
\email{a.jacquier@imperial.ac.uk}

\author{Aitor Muguruza}
\address{Kaiju Capital Management and Department of Mathematics, Imperial College London}
\email{aitor.muguruza-gonzalez15@imperial.ac.uk}

\date{\today}
\thanks{AJ acknowledges financial support from the EPSRC grant EP/T032146/1.}

\maketitle
\begin{abstract}
We develop a new analysis for portfolio optimisation with options, tackling the three fundamental issues with this problem: asymmetric options' distributions, high dimensionality and dependence structure.
To do so, we propose a new dependency matrix, 
built upon conditional probabilities between options' payoffs, and show how it can be computed in closed form given a copula structure of the underlying asset prices.
The empirical evidence we provide highlights that this approach is efficient, fast and easily scalable to large portfolios of (mixed) options.
\end{abstract}

\tableofcontents

%%%%%%%%%%%%%%%%%%%%%%%%%%%%%%%%%%%%%%%%%%%%%%%%%%%
%%%%%%%%%%%%%%%%%%%%%%%%%%%%%%%%%%%%%%%%%%%%%%%%%%%
\section{Introduction}
The problem of building a portfolio of stocks,  maximising expected returns for a given level of risk, was tackled by Markowitz as far back as 1952 in his seminal paper  on portfolio selection~\cite{markowitz}. 
While this was the foundation stone for modern portfolio theory, subsequently used widely in the financial industry, and worthy of a Nobel Prize, 
it is not a universal mechanism
as it fails to tackle very asymmetric or fat-tailed distributions; 
Post-modern portfolio theory~\cite{rom1994post} was an attempt to improve it, in particular recognising that Equity asset returns are not symmetric and have fat tails.
Quite remarkably however, there seems to be little literature on how to construct optimal portfolios, not of single assets, but of European options. 
The main issue is that these options display highly asymmetric returns distributions,
making previous theories shaky. 
Investing in options instead of single assets
is a more high-risk strategy because of their `all or nothing' payout structure. 
However, purchasing long options offers advantages that stocks do not, 
in particular limited downside and high leverage. Deep out-of-the-money (OTM) options can be purchased at a fraction of the price of the underlying stock, and given a large enough movement in the underlying, large returns are at hand.

The key difference between investing in options and in stocks is the risk, in the former case, 
to lose it all, even without extreme events like default.
Unfortunately, this risk is not well represented by the variance of an option's distribution:
whereas high variance in a stock returns distribution suggests high likelihood of going either up and down,
for deep OTM options, the loss is limited to$-100\%$ and high variance is mostly due to the right tail of the distribution. 
Therefore high variance may reflect a high possibility of large positive returns, which is certainly not a risk. 
In fact, low variance may reflect an option's distribution with the majority of its mass centered around $-100\%$ returns and a high risk of an investor losing all their money. 
Minimising the variance of the portfolio in a Markowitz way may only be making things worse.
Let us illustrate the main issue with an ad-hoc example, not fully realistic, yet simple and sensible.
For some fixed time horizon $T>0$, 
consider two independent Bernoulli random variables~$X_T$ and~$Y_T$, taking values in $\{0,1\}$, 
so that, for example, $(X_t)_{t\in [0,T]}$ and $(Y_t)_{t\in [0,T]}$ represent the prices of 
two Digital options, 
and~$X_T,Y_T$ their respective payoffs.
We assign probabilities to the outcomes as follows:
$$
\PP(X_T = 0) = 0.5,\qquad
\PP(X_T = 1) = 0.5,\qquad
\PP(Y_T = 0) = 0.9,\qquad
\PP(Y_T = 1) = 0.1,
$$
so that 
$\EE[X_T]= .5$,
and
$\EE[Y_T] = .1$,
$\VV[X_T] = 0.25$
and $\VV[Y_T] = 0.09$.
It is clearly tempting to say that the~$X$ option
is less risky than~$Y$. 
Construct now, at time zero, the portfolio
$\Pi_{0} := w_{X}X_{0} + w_{Y}Y_{0}$
such that $w_{X}+w_{Y}=1$.
The variance  
$\VV[\Pi_T] = w_{X}^2\VV[X_{T}] + w_{Y}^2\VV[Y_{T}]$
of $\Pi_T$ is thus minimised if 
$$
w_{X} = \frac{\VV[Y_T]}{\VV[X_T] + \VV[Y_T]} = 0.265
\qquad\text{and}\qquad
w_{Y} = \frac{\VV[X_T]}{\VV[X_T] + \VV[Y_T]} = 0.735.
$$
Contrary to intuition, Markowitz's criterion selects more options~$Y$ than options~$X$.
This seemingly contradictory result can easily be explained by the fact that it selects options with mass centered at $-100\%$ returns, increasing rather than reducing the risk.

%%%%%%%%%%%%%%%%%%%%%%%%%%%%%%%%%%%%%%%%%%%%%%%%%%%%

The picture is complicated even further when dealing with multivariate option payout distributions. 
Underlying assets, at least on Equity, have asymmetric and fat-tailed distributions, along with a covariance matrix linking them together.
This in turn implies highly asymmetric multivariate distributions for the options,
and explains why Modern Portfolio Theory struggles with option portfolios. 
A natural extension is to optimise instead for high skew and low kurtosis, or using a CRRA utility function. 
While this may work for small portfolios,
the fact that skew and kurtosis are rank~$3$ and~$4$ tensors respectively makes the problem quickly run into the curse of dimensionality and becomes unfeasible.

To tackle these issues, two approaches 
have been followed: 
the first optimises a utility function taking into account higher-order moments, such as CRRA, 
while the second optimises with respect to Greek preferences. 
In general, these papers show good results on metrics like  annualised Sharpe ratio, but most seem incapable of dealing with high-dimensional options portfolios. 
For example, 
Coval and Shumway~\cite{coval}, Saretto and Santa-Clara~\cite{san}, and Driessen and Maenhout~\cite{dri}
demonstrated that short positions in crash-protected, delta-neutral straddles have high Sharpe ratios. 
Coval and Shumway ~\cite{coval} also showed that selling naked Puts can offer good levels of returns, although this is a very risky strategy to use in practice, due to the unlimited downside of selling naked options. 
However interesting these strategies are, 
they unfortunately offer little insight on building portfolios with multiple options. 
Recently Santa Clara and Faias~\cite{Faias2017OptimalOP} 
optimised a power utility function numerically, using simulated option returns and realistic transaction costs, 
with a reported Sharpe ratio of~$0.82$, 
but only trades the risk-free bond and four Call and Put options, all on the same S\&P 500 index tracker. This approach can be adjusted to consider multiple options and stocks, but would unfortunately not scale well for large portfolios.   
Riaz and Wilmott~\cite{Riaz} proposed ways to profit from mispriced options, hedged with implied volatility, but, 
based on PDE techniques, the curse of dimensionality is a major bottleneck.
In 2013 Eraker~\cite{Eraker} extended the  mean-variance framework with a parametric model
of stochastic volatility to select weights on three options: at-the-money straddles,
OTM Puts and OTM Calls. 
His annualised Sharpe ratio reaches~$1$, 
but the optimal portfolio almost exclusively shorts Put options and goes Long on Call options, 
a very risky strategy during market meltdowns.     
Driessen and Maenhout~\cite{Dti} empirically analysed portfolios with one stock and several options, optimising a CRRA utility function,
concluding that shorting OTM options is the best line of practice, but again failed short of handling the large portfolio case.

A related problem is that of pricing basket options. 
Classical methods rely on assuming 
some dynamics for the underlying and specifying a correlation matrix between the driving noises.
In fact, an extension of the classical Dupire local volatility model~\cite{dupire1994pricing}
for single-stock European options
was developed in~\cite{galichon2006modelling}
and~\cite{langnau2009introduction}, ensuring exact calibration to a given set of basket options.
Copulas were used in~\cite{cherubini2002bivariate}
and~\cite{van2005bivariate} to evaluate options on two underlyings.
and Bernard and Czado~\cite{bernard2013multivariate} proposed GARCH dynamics
for the underlyings, 
where the dependence structure is defined through copulas.
We are not however considering path-dependent options here, 
and we instead borrow this copula idea indeed solely for the marginal distributions of the underlyings,
our goal being to analyse the induced dependence structure of the options' portfolio.

The closest approach to ours is the study by Malamud~\cite{malamud}, who replaces Markowitz' covariance matrix with a `Greek efficient' matrix. 
He achieves high Sharpe ratios but also very high kurtosis, which is somewhat to be expected in OTM option portfolios. 
This is done on options for three major stock index tracker funds. Although not demonstrated theoretically the approach, like ours, 
should scale well because the `Greek efficient' matrix is two-dimensional.

Our approach builds on~\cite{malamud} and aims at solving the high-dimensional issue while taking into account asymmetry and fat-tailed distributions of option prices.
We introduce a dependency matrix, based 
on copulas to create bivariate dependency measures
between options' payoffs
and use it to replace Markowitz' covariance matrix in the optimisation problem. 
In Section~\ref{sec:Copulas}, we briefly review copulas and show how to select the right ones via maximum likelihood.
Section~\ref{sec:DependPtf} is the key part,
where we introduce the new dependency matrix for portfolios of options and show in particular that, given a copula structure between assets, every term in the matrix can be computed in closed form.
In Section~\ref{sec:PtfOptim}, we describe the optimisation procedure and perform a detailed backtest analysis of it in Section~\ref{sec:Results}.

%%%%%%%%%%%%%%%%%%%%%%%%%%%%%%%%%%%%%%%%%%%%%%%%%%%%%%%%%%%%%%%%%%%%%%%%%%%%%%%%%%%%%%%%%%%%%%%%%%%%%%%%
\section{Modelling option dependence with copulas}\label{sec:Copulas}

Copulas are a common tool to model multivariate distributions. 
They take marginal distributions from multiple random variables as inputs and are flexible enough to model complex non-linear dependence structures.
We review here their fundamental properties and mention a few which will be useful later to characterise  the dependence between OTM options in a portfolio.

%%%%%%%%%%%%%%%%%%%%%%%%%%%%%%%%%%%%%%%%%%%%%%%%%%%%%%%%
\subsection{Fundamentals of copulas}
For a given random vector~$\Xb \in \RR^n$, 
its cumulative distribution function (cdf) is the map
$\Fb:\RR^n\to[0,1]$ defined as
$$
\Fb(\xx)=\PP\left(X_{1} \leq x_{1}, \ldots, X_{n} \leq x_{n}\right),
\qquad\text{ for any }\xx = (x_{1}, \ldots, x_{n}) \in \RR^n.
$$
Assuming that the coordinates $\{F_{i}:=\PP(X_i\leq \cdot)\}_{i=1,\ldots,n}$ of~$\Fb$
are continuous functions, 
the random vector $(F_1(X_1), \ldots, F_n(X_n))$
has uniform marginal distributions on $[0,1]$.

\begin{definition}
A map $\Cop:[0,1]^n \to [0,1]$ is called an $n$-dimensional copula 
if it is the joint cdf of an $n$-dimensional random vector on 
$[0,1]^{n}$ with uniform marginals.
\end{definition}
The key result about copulas is due to Sklar~\cite{sklar}:
\begin{theorem}\label{thm:Sklar}
For any multivariate cdf~$\Fb$, 
there exists a copula~$\Cop$ such that
$$
\Fb(\xx) = \Cop\Big(F_1(x_1), \ldots, F_n(x_n)\Big),
\qquad\text{ for all }\xx \in \RR^n.
$$
\end{theorem}
If each~$F_i$ admits a density~$f_i$, 
we can express the copula density as
\begin{equation}\label{eq:CopulaDensity}
\cop(\uu) := \partial_{u_{1}, \ldots, u_{n}}\Cop(\uu)
 = \frac{f\left(\xx\right)}{\prod_{i=1}^{n} f_{i}\left(x_{i}\right)},
\qquad\text{ for any }\uu \in [0,1]^n,
\end{equation}
where $\xx$ and $\uu$ are related via \begin{equation}\label{eq:uxRelation}
u_i = F_i(x_i), 
\qquad\text{for each }i=1,\ldots, n.
\end{equation}

The so-called Fr\'echet–Hoeffding bounds
will be useful later:
\begin{equation}\label{eq:FrechetHoeffding}
\max \left\{\sum_{i=1}^{n} u_{i}+1-n, 0\right\}
\leq \Cop(\uu)
\leq \min \left\{u_{1}, \ldots, u_{n}\right\},
\qquad\text{ for any }\uu \in [0,1]^n,
\end{equation}
It can be shown that the upper bound is always sharp, in the sense that the $\min$ function
defines a copula.
However this is not true for the lower bound unless $n=2$.

%%%%%%%%%%%%%%%%%%%%%%%%%%%%%%%%%%%%%%%%%%%%%%%%%%%%%%%%
\subsection{Examples of copulas}

Many classes of copulas have been used to model multivariate distributions. 
Choosing the right one has often been influenced by pre-perceived ideas about the structure of data dependence, 
although, we will see later in Section~\ref{sec:FittingData}
how this choice can be made solely through a data-driven procedure.
In our short review here, 
we concentrate on classes of copulas widely used in Finance, 
discussing their merits and disadvantages, 
and cast a closer look at specific ones 
we shall use in our analysis.

%%%%%%%%%%%%%%%%%%%%%%%%%%%%%%%%%%%%%%%%%%%%%%%%%%%%%%%%
\subsubsection{Elliptical copulas}

In the early 2000s, elliptical copulas,
a type of implicit copulas,
were widely used in Finance.
The name comes from their derivation from elliptical multivariate distributions~\cite{Cambanis} via Sklar's Theorem~\ref{thm:Sklar}.

\begin{definition}
A random vector $\Xb \in \RR^n$ has an elliptical distribution
if there exist a function $\psi:\RR\to\RR$, 
a location parameter $\mum\in\RR^n$ and
a non-negative-definite square real matrix~$\Sigmab$
such that
$$
\EE\left[\exp\left\{\I\uu^\top \Xb\right\}\right] = \psi\Big(\uu^\top\mum + \uu^\top\Sigmab\uu\Big),
\qquad\text{for any }\uu\in\RR^n.
$$
\end{definition}
The simplest example is the Gaussian copula,
where~$\Xb$ is centered Gaussian with covariance matrix~$\Sigmab$, so that
$$
\Cop_{\Sigmab}^{\mathrm{G}}(\uu)
 = \Fb_{\Xb}\left(\Phi^{-1}(u_{1}), \ldots, \Phi^{-1}(u_{n})\right),
\qquad\text{for any }\uu\in\RR^n,
$$
where~$\Phi$ is the one-dimensional standard Gaussian cumulative distribution function
and~$\Fb_{\Xb}$ the cdf of~$\Xb$ .
In general, elliptical copulas are relatively easy to fit 
given an estimate of the covariance matrix. 
They are symmetric and tail dependence, essential for modelling option portfolios, can be adapted to the data.
They became widespread in Finance to model default correlations across large portfolios of CDOs~\cite{Li};
however they significantly underestimate tail dependence, leading them to become `the Number~1 formula that killed Wall Street'~\cite{Salmon} following the 2008 crisis.
Other elliptical copulas may alleviate this, 
but are usually not flexible enough to deal with asymmetric or extreme dependence~\cite{CHICHEPORTICHE_2012}, 
a common feature in option portfolios. 

%%%%%%%%%%%%%%%%%%%%%%%%%%%%%%%%%%%%%%%%%%%%%%%%%%%%%
\subsubsection{Archimedean copulas}
Archimedean copulas~\cite[Chapter 4]{nelsen2007introduction} can handle a large variety of different symmetric and asymmetric dependence structures and have thus 
seen increased interest in risk management and natural disaster modelling. 
They are however hard to interpret and the right choice is not (yet) supported by strong arguments~\cite{CHICHEPORTICHE_2012}.

\begin{definition}
A monotone decreasing function $\psi:[0,\infty) \to [0,1]$ with $\psi(0)=1$ and $\lim_{z\uparrow\infty}\psi(z)=0$.
is called a generator.  
An Archimedean copula~$\Cop$ can be represented via a generator~$\psi$ as
$$
\Cop(\uu)=\psi^{-1}\left(\psi\left(u_{1}\right)+\cdots+\psi\left(u_{n}\right)\right),
\qquad\text{for all } \uu=(u_1,\ldots,n_n)\in\RR^n.
$$
\end{definition}

Provided that~$\psi$ is smooth enough, then from~\eqref{eq:CopulaDensity} we have an analytical expression for the copula density function $\cop(\cdot)$, which will prove useful for Maximum Likelihood Estimation:

\begin{equation}\label{eq:copArchi}
\cop(\uu)=\frac{\psi\left(\psi^{-1}\left(u_{1}\right)+\cdots+\psi^{-1}\left(u_{n}\right)\right)}{\prod_{j=1}^{n} \psi'\left(\psi^{-1}\left(u_{j}\right)\right)},
\end{equation}
for all $\uu=(u_1,\ldots,n_n)\in\RR^n$
such that $\psi'\left(\psi^{-1}\left(u_{j}\right)\right)\ne 0$.
Archimedean copulas are one-parameter models and are most commonly used to model bivariate distributions.
These $n$-dimensional models can then be linked together with Vine copulas, as discussed below, to model returns and dependence in higher dimensions. 
For the rest of the paper, unless explicitly stated we will be focusing on two-dimensional Archimedean copulas.
Should we model~$m$ stocks, we would have $\frac{m(m-1)}{2}$ bivariate Archimedean copulas and the same number of parameters,
which scales as $\mathcal{O}(m^2)$, 
so unless the portfolio gets very large, we should not suffer from issues with dimensionality.
There are many Archimedean copulas, which can even be mixed together to form new copulas to model all kinds of interesting dependence structures~\cite[Chapter 4]{nelsen2007introduction}. 
We also refer to~\cite{nagler2019model} for further examples of parametric pair-copulas in the context of portfolio risk.
Table~\ref{tab:ArchimCopulas} gathers three important copulas, on which we shall focus from now on.
For each of them, a parameter~$\theta$ decides the strength of dependence between the two variables,
and will be analysed later when fitting to data.

\begin{table}
\begin{equation*}
\begin{array}{|c|c|c|c|}
\hline & \text { Generator $\psi$ } & \text { Inverse $\psi^{-1}$ } & \text { Bivariate copula $C_{\psi}$ } \\
\hline \begin{array}{c}
\text {Clayton}\\
\text{}
\end{array}
& t^{-\theta}-1
& (1+s)^{-\frac{1}{\theta}}
& \text{max}\{\left(u^{-\theta}+v^{-\theta}-1\right),0\}^{-\frac{1}{\theta}}\\
\hline \begin{array}{c}
\text{Frank}\\
\text{}
\end{array}
& \ln \left(\frac{\E^{\theta t}-1}{\E^{\theta}-1}\right)
& -\frac{1}{\theta} \log\left(1+\E^{s}\left(\E^{-\theta}-1\right)\right)
& -\frac{1}{\theta}\log\left(1+\frac{\left(\E^{-\theta u}-1\right)\left(\E^{-\theta v}-1\right)}{\E^{-\theta}-1}\right)\\
\hline \begin{array}{c}
\text{Gumbel}\\
\text{}
\end{array}
& (-\ln t)^{-\theta}
& \exp \left(-s^{\frac{1}{\theta}}\right)
&\exp \left\{-\left((-\log (u))^{\theta}+(-\log (v))^{\theta}\right)^{\frac{1}{\theta}}\right\} \\
\hline
\end{array}
\end{equation*}
\caption{Examples of Archimedean copulas.}
\label{tab:ArchimCopulas}
\end{table}

%%%%%%%%%%%%%%%%%%%%%%%%%%%%%%%%%%%%%%%%%%%%%%%%
\subsubsection{Empirical copulas}
Instead of specifying a parametric class of copulas, 
one can instead characterise them empirically from data.
For any $k=1,\ldots, n$ and a sample of size~$m$,
consider the marginal empirical cdf
$$
\widehat{F}_{m,k}(x)
 := \frac{1}{m + 1} \sum_{i=1}^{m} \ind_{\left\{X_{k}^{(i)} \leq x\right\}},
 \qquad\text{for any }x\in\RR,
$$
where $\{X_{k}^{(i)}\}_{i=1,\ldots,m}$ represents $m$ realisations of the random variable~$X_{k}$.
We can therefore write the empirical copula as
\begin{equation}\label{eq:EmpiricalCopula}
\widehat{\Cop}_{n}(\uu)
 := \frac{1}{n} \sum_{i=1}^{n}\ind_{\left\{\widehat{F}_{m,1}\left(x_{1}\right) \leq u_{1}, \ldots, \widehat{F}_{m,n}\left(x_{n}\right) \leq u_{n}\right\}},
\qquad\text{for any }\uu\in\RR^n,
\end{equation}
where~$\xx$ and~$\uu$ are again defined via~\eqref{eq:uxRelation} using~$\widehat{F}_i$ instead of~$F_i$.
The empirical copula~$\widehat{\Cop}_n$ is known to converge weakly towards the true copula~$\Cop$
as the number of data points~$n$ increases, 
as long as~$\Cop$ has continuous partial derivatives~\cite[Lemma 7]{Cop_Convergence}.
Unfortunately, lack of large options datasets often makes it hard to use, 
but it will nevertheless come handy as a goodness-of-fit test for other Archimedean copulas.

%%%%%%%%%%%%%%%%%%%%%%%%%%%%%%%%%%%%%%%%%%%%%%%%%%%%%%%%%
\subsubsection{Vine copulas}
Although Archimedean copulas seem to offer a flexible way to model different dependence structures between pairs, 
it is not clear how to use them for high-dimensional dependency structures. 
Vine copulas provide a good way to model each pair of returns using Archimedean copulas and join them together in a vine, 
allowing highly complex dependency structures in high dimensions. 
Due to this flexibility, Vine copulas have become of increasing interest in practice, 
finding uses in portfolio optimisation and risk management~\cite{vines_RM}.
Joe~\cite{IFM} and Bedford and Cooke~\cite{bedford2002vines} introduced vine copulas and highlighted their flexibility over classical copulas. They defined a so-called R-vine as follows:
\begin{definition}
An R-vine with $n$ elements is a sequence of tree structures $(T_1, \ldots, T_n)$ such that
\begin{itemize}
\item $T_1$ is a tree with nodes $N_1 = \{1, \ldots, n\}$ and a set of edges~$E_1$;
\item For each $i=2, \ldots, n-1$, $T_i$ with nodes $N_i = E_{i-1}$ and edge set~$E_i$;
\item For each $i=2, \ldots, n-1$ and $\{\mathrm{a}=\{a_1,a_2\}, \mathrm{b}=\{b_1,b_2\}\} \in E_i$ then the cardinal of $\mathrm{a}\cap\mathrm{b}$
must be equal to~$1$.
\end{itemize}
\end{definition}
Recall that~\cite[Definition~2.4]{bedford2002vines} a tree is an acyclic graph, 
so that~$n$ nodes can only result in a maximum of $n-1$ edges.
R-vines are very flexible and possess many interesting and useful properties,
detailed in~\cite{bedford2002vines, kurowicka2010optimal}.
They also encompass two particular structures, D-vines and C-vines.
The latter are characterised by their star structures (at least one node in each tree has maximal degree)
while D-vines are a subclass of R-vines such that the first tree~$T_1$
has nodes with degree at most equal to two (recall that the degree of a node corresponds to the number of neighbours in the tree).
As shown in~\cite{aas2009pair}, and developed in~\cite{hofmann2010assessing, kraus2017d},
D-vines have enhanced statistical modelling properties, which make them easier and more suitable for applications.
We shall thus focus on these from now on.
The underlying idea of D-vine copulas is that the multivariate probability density can be constructed as a product of smaller bivariate copula densities and marginal density functions~\cite{Vine,Vine_d}:
\begin{equation}
f(\xx) = \prod_{j=1}^{n-1} \prod_{i=1}^{n-j} c_{i,i+j |(i+1), \cdots,(i+j-1)} \cdot \prod_{k=1}^{n} f_{k}(x_{k}),
\qquad \text{for all }\xx \in \RR^n,
\end{equation}
where the coefficients 
$c_{i,i+j |(i+1), \cdots,(i+j-1)}$ represent the conditional copula densities for the pair of variables $(i,i+j)$, given the variables indexed in between~\cite{Vines_copula_pairs}.
We refer the interested reader to~\cite{Vine} for a thorough review of Vine copulas.

%%%%%%%%%%%%%%%%%%%%%%%%%%%%%%%%%%%%%%%%%%%%%%%%%%%
%%%%%%%%%%%%%%%%%%%%%%%%%%%%%%%%%%%%%%%%%%%%%%%%%%%
\subsection{Fitting copulas to data}\label{sec:FittingData}
Given a class of copulas, we show how to apply Maximum Likelihood Estimation (MLE)
to fit the copula parameters to data.
This is not the only way to fit copulas, and a detailed survey of goodness-of-fit tests can be found in~\cite{genest2009goodness}. 
Because of the availability of the density function, we find MLE suitable enough for our analysis though.
We assume that the joint probability density function $f(\Xb;\thb)$ of the $n$-dimensional random variable~$\Xb$ exists and is dependent on the vector of parameters $\thb = (\theta,\alb)$, where~$\theta$ is the  copula parameter 
and~$\alb$ are the parameters associated with~$\Xb$.
Given~$m$ independent observations
$\Xb^{(1)}, \ldots, \Xb^{(m)}$
of~$\Xb$, we defined the likelihood function as
\begin{equation*}
\Ll_{m}(\thb) := \prod_{j=1}^{m} 
f\left(\Xb^{(j)} ; \thb\right).
\end{equation*}
From Sklar's Theorem~\ref{thm:Sklar}, we can rewrite the joint density function $f(\Xb;\thb)$ as a product of the copula density function $\cop(\uu;\theta)$ and the marginal density functions $f_i(x_i;\alpha_i)$ as in~\cite{choros2010copula}:
$$
f\left(\Xb;\thb\right) = c\left(\uu;\theta \right) {\prod_{i=1}^{d} f_{i}\left(X_{i};\alpha_i \right)}.
$$
Maximising the log-likelihood function 
is equivalent to maximising the likelihood function~$\Ll_m$.
We can then split the log-likelihood function into two parts as
$$
l_m(\thb) 
:= \log \Ll_m(\thb)=  \sum_{j=1}^{m} \log\left(\cop(\uu^{(j)};\theta)\right)
+ \sum_{i = 1}^{n} \sum_{j = 1}^{m} \log\left(f_i\left(X_i^{(j)};\alpha_i\right)\right),
$$
where the second term on the right-hand side
corresponds to the log-likelihood under the independence assumption.
We then estimate the parameters $\thb = (\theta,\alb)$ via
$$
\widehat{\thb} := \argmax _{\thb}
l_m(\thb).
$$
Alternatively one can estimate the parameters using the method of inference functions~\cite{IFM},
which is a two-step process:
it first finds $\{\alpha_i^{\text{IFM}}\}_{i=1,\ldots,n}$ for all the marginal distributions, 
maximising the likelihood function for the individual marginals $\{f_i(X_i;\alpha_i)\}_{i=1,\ldots,n}$;
it then determines~$\theta^{\text{IFM}}$ by maximising the total log-likelihood function using the estimates $\{\alpha_i^{\text{IFM}}\}_{i=1,\ldots,n}$~\cite{choros2010copula}. 
A simpler method for estimation is with semi-parametric estimation. 
Instead of estimating the marginal density functions~$\{f_i\}_{i=1,\ldots,n}$ parametrically, we can use the empirical data to estimate the marginal cdfs~$\{\widehat{F}_i(X_i)\}_{i=1,\ldots,n}$. 
Then we only have a one-step optimisation procedure to find the best estimate
$$
\widehat{\thb} = \argmax _{\thb} \sum_{j=1}^{m} \log\left(\cop\left(\widehat{F}_{1}\left(X_{1}^{(j)}\right), \ldots, \widehat{F}_{n}\left(X_{n}^{(j)}\right) ; \theta\right)\right).
$$
It is shown in~\cite{Semi_par} that this semi-parametric estimator is consistent for almost all Archimedean copulas, under some copula regularity conditions. 
The one-step MLE estimation approach is likely to yield the "best" estimate for~$\theta$, 
but it also tends to be slow, 
without exact solution, requiring a computationally costly optimisation procedure. 
The IFM method and semi-parametric methods are much quicker to solve, and yield consistent estimators~\cite{choros2010copula}. 
For this reason, we shall later use IFM or the semi-parametric approach.
Once the optimal~$\widehat{\thb}$ has been found
for each of the suggested copulas,
we need to decide which one is the best for a particular pair of stock returns.
We do so by picking the copula that most closely resembles the empirical copula~\eqref{eq:EmpiricalCopula}. 
More precisely, given~$n$ (Archimedean) copulas $\Cop_i(\thb_i)$, we compute the~$L^2$ distance
(other norms could be chosen)
$$
\left\|\widehat{\Cop} - \Cop_{i}\right\|_{L^{2}}
 := \left(\int_{[0,1]^{2}}\left|\widehat{\Cop}(\uu) - \Cop_{i}(\uu)\right|^{2} \mathrm{~d} \uu\right)^{1/2},
$$
and select the copula with the smallest error.

%%%%%%%%%%%%%%%%%%%%%%%%%%%%%%%%%%%%%%%%%%%%%%%%%%%
\section{Constructing dependence structure for options portfolios}\label{sec:DependPtf}

%%%%%%%%%%%%%%%%%%%%%%%%%%%%%%%%%%%%%%%%%%%%%%%%%%%
\subsection{Tail dependence with copulas}
As mentioned previously, the key issue with options portfolios is their dependence structure, 
which occurs mostly in the tails
where out-of-the-money options generate positive payoffs.
We thus introduce tail dependence coefficients
to measure dependence of random variables in their tail distributions.
We fix for now two continuous random variables~$X_1$ and~$X_2$ and consider a copula~$\Cop$ between them.
We are not assuming that their distributions functions are invertible, 
so the left-arrow inverse notation~$\FInv$ shall always refer to the generalised inverse, which is well defined.

\begin{definition}
The upper and lower tail dependence coefficients 
are defined as
\begin{equation}\label{eq:TailDepCoef}
\begin{array}{rl}
\lambda_{U}\left(X_{1}, X_{2}\right)
& \displaystyle := \lim _{u \uparrow 1} \PP\left(X_{2}>\FInv_{2}(u) \Big| X_{1}>\FInv_{1}(u)\right),\\
\lambda_{L}\left(X_{1}, X_{2}\right) 
& \displaystyle := \lim _{u \downarrow 0} \PP\left(X_{2} \leq \FInv_{2}(u) \Big| X_{1} \leq \FInv_{1}(u)\right).
\end{array}
\end{equation}
\end{definition}
Let us state a few properties of these coefficients:
\begin{proposition}\label{prop:TailDepCoef}
Both~$\lambda_{L}$ and~$\lambda_{U}$
are symmetric in their arguments and
\begin{equation}\label{TailCoefCop}
\lambda_{U}(X_1,X_2)=\lim _{u \uparrow 1} \frac{1-2 u+\Cop(u, u)}{1-u} 
\qquad\text {and}\qquad
\lambda_{L}(X_1,X_2)=\lim _{u \downarrow 0} \frac{\Cop(u,u)}{u}.
\end{equation}
\end{proposition}

\begin{proof}
To prove the symmetry property, we first write
\begin{align*}
\lambda_{U}\left(X_{1}, X_{2}\right)
 & = \lim _{u \uparrow 1} \PP\left(X_{2}>\FInv_{2}(u) \Big| X_{1}>\FInv_{1}(u)\right)\\
 & = \lim _{u \uparrow 1} \frac{\PP\left(\{X_{2}>\FInv_{2}(u)\} \cap \{X_{1}>\FInv_{1}(u)\}\right)}{ \PP\left(X_{1}>\FInv_{1}(u)\right)}\\
 & = \lim _{u \uparrow 1} \frac{\PP\left(\{X_{2}>\FInv_{2}(u)\}\cap \{X_{1}>\FInv_{1}(u)\}\right)}{1 - u}\\
 & = \lim _{u \uparrow 1} \frac{\PP\left(\{X_{2}>\FInv_{2}(u)\}\cap \{X_{1}>\FInv_{1}(u)\}\right)}{ \PP\left(X_{2}>\FInv_{2}(u)\right)}\\
 & = \lim _{u \uparrow 1} \PP\left(X_{1}>\FInv_{1}(u) \Big| X_{2}>\FInv_{2}(u)\right)
  = \lambda_{U}\left(X_{2}, X_{1}\right).
\end{align*}
Similarly
\begin{align*}
\lambda_{L}\left(X_{1}, X_{2}\right) 
 & = \lim _{u \downarrow 0}
\PP\left(X_{2} \leq \FInv_{2}(u) \Big| X_{1} \leq \FInv_{1}(u)\right)\\
 & = \lim _{u \downarrow 0} 
\PP\left(X_{1} \leq \FInv_{1}(u) \Big| X_{2} \leq \FInv_{2}(u)\right)
 = \lambda_{L}\left(X_{2}, X_{1}\right).
\end{align*}
The representations~\eqref{TailCoefCop}
come from~\cite[Theorem~5.4.2]{nelsen2007introduction}.
\end{proof}

The tail dependence coefficients range from~$0$ for no dependence to~$1$ for complete dependence.
Ideally a well-diversified portfolio of options should penalise high tail dependence, 
in particular for deep out-of-the-money options where extremes of the distributions dominate. 
For example, for three deep OTM Call options on three stocks, with two of them having upper tail~$\lambda_U$ close to~$1$ and  one of them being close to~$0$, it may be wise to avoid a linear combination of the first two since they will only payout at the same time.
However tail dependence has many limitations as a dependence measure that make it not fully suitable for measuring dependence in options portfolio. 
First it only compares dependence at the very end of both tails, which is hard to map to observed moneynesses. Second consider two Call options on the same stock
with different strikes; they will have tail dependence equal to~$1$ no matter how far apart the strikes are. This lack of flexibility and interpretability means that one may be unable to use them as dependency measures. 

%%%%%%%%%%%%%%%%%%%%%%%%%%%%%%%%%%%%%%%%%%%%%%%
\subsection{Conditional probabilities}

Inspired by the tail dependence coefficients~\eqref{eq:TailDepCoef}, 
we introduce a new metric 
given by conditional probabilities that an option pays out given that another pays out.
The two options have the same maturity
and bivariate copulas can be used to generate explicit formulas.
Consider two European options~$\Op^1$ and~$\Op^2$ with associated strikes~$K_1$ and~$K_2$, 
on two stocks~$S^1$ and~$S^2$ with continuous distributions, with the same fixed maturity~$T$ and copula~$\Cop$ between them.
For $i=1,2$, we shall further write~$\Op^i_+$ to indicate the event that~$\Op^i$ has strictly positive payoff, and define the conditional probability 
\begin{equation}\label{eq:muFormula}
\mu(\Op^1,\Op^2) :=\PP\left(\Op^1_+\Big|\Op^2_+\right), 
\end{equation}
so that
with~$\Cf$, $\Pf$ denoting Calls and Puts, we have
\begin{align*}
\mu(\Cf^1,\Cf^2) & :=
\PP\left(\Cf^1_+\Big|\Cf^2_+\right)
= \PP\left(\left(S_T^1 - K_1\right)_{+} > 0 \Big| \left(S_T^2 - K_2\right)_{+} > 0 \right),\\
\mu(\Pf^1,\Pf^2) & :=
\PP\left(\Pf^1_+\Big|\Pf^2_+\right)
= \PP\left(\left(K_1-S_T^1\right)_{+} > 0 \Big| \left(K_2-S_T^2\right)_{+} > 0 \right).
\end{align*}
Beyond standard Calls and Puts, 
many other European options can be handled in our framework with a little bit of algebra.
We consider for example the (slightly more involved) case of a Strangle, 
which combines a Put and a Call (on the same underlying) with the strike of the latter larger than the strike of the former.
We let~$\Sf^i$ ($i=1,2$) denote the Strangles,
with their associated strikes $K^i< \widetilde{K}^i$, so that the positive payoff of~$\Sf^i$ can be written as
$$
\Sf^i_+ = \left\{\left(S_T^i - \widetilde{K}_i\right)_{+}  > 0\right\}
\cup 
\left\{\left(K_i - S_T^i\right)_{+} > 0\right\}
 = \Cft^i_+ \cup\Pf^i_+,
$$
where~$\Cft^i$ corresponds to a Call with strike~$\widetilde{K}^i$.
Let us first state and prove the following lemma, linking payoff probabilities to the copula between the two stock prices.

\begin{lemma}\label{lem:AllProba}
With $(u_1,u_2) = (F_1(K_1), F_2(K_2))$, we have
\begin{align}
\PP\left(\Cf^1_+\cap\Cf^2_+\right)
&= 1 - u_1 - u_2 + \Cop(u_1,u_2),\label{eq:Ct1pCt2p}\\
\PP\left(\Pf^1_+\cap\Pf^2_+\right)
 & = \Cop(u_1,u_2),\label{eq:P1pP2p}\\
\PP\left(\Cf^1_+\cap \Pf^2_+\right)
 & = u_2 -  \Cop(u_1, u_2)\label{eq:Ct1pP2p}\\
\PP\left(\Pf^1_+\cap\Cf^2_+\right)
 & = u_1 -  \Cop(u_1, u_2),\label{eq:P1pCt2p}\\
\PP\left(\Cf^2_+\cup\Pf^2_+\right)
 & = 1 - u_2 + u_2,\label{eq:Ct2pP2p}.
\end{align}
\end{lemma}

\begin{proof}[Proof of Lemma~\ref{lem:AllProba}]
By definition, we can write
\begin{align*}
\PP\left(\Cf^1_+\cap\Cf^2_+\right)
 & = \PP\left(\{(S_T^1 - K_1)_{+}  > 0\}\cap\{(S_T^2 - K_2)_{+} > 0\}\right)\nonumber\\
&= 1 - \PP\left(S_T^1 > K_1)\right)
- \PP\left(S_T^2 > K_2)\right)
+ \PP\left(\{S_T^1 < K_1)\}\cap\{(S_T^2 < K_2)\}\right)\nonumber\\
&= 1 - F_1(K_1) -  F_2(K_2) + F_{12}(K_1),K_2)\nonumber\\
&= 1 - u_1 - u_2 + \Cop(u_1,u_2),
\end{align*}
where $F_{12}$ is the joint cdf.
Now,
$$
\PP\left(\Pf^1_+\cap\Pf^2_+\right)
 = \PP\left(\{(K_1 - S_T^1)_{+}  > 0)\}\cap\{(K_2 - S_T^2)_{+} > 0\}\right)
= F_{12}(K_1,K_2)
 = \Cop(u_1,u_2),
$$
and
\begin{align*}
\PP\left(\Cf^1_+\cap \Pf^2_+\right)
 & = \PP\left(\{(S_T^1 - K_1)_{+}  > 0\}\cap \{(K_2 - S_T^2)_{+} > 0\}\right)
\\
 & = \PP\left(K_2 < S_T^2\right)
- \PP\left(\left\{(S_T^1 < K_1\right\} \cap \left\{S_T^2 < K_2\right\}\right)\\
 & = F_2(K_2) - F_{12}\left(K_1,K_2\right)
 = u_2 -  \Cop(u_1, u_2).
\end{align*}
Likewise,
\begin{align*}
\PP\left(\Pf^1_+\cap\Cf^2_+\right)
 & = \PP\left(\{(K_1 - S_T^1)_{+}  > 0\} \cap\{(S_T^2 - K_2)_{+} > 0\}\right)\\
& = \PP\left(K_1 < S_T^1\right)
- \PP\left(\{S_T^2 < K_2\}\cap \{S_T^1 < K_1\}\right)\\
& = F_1(K_1) - F_{12}\left(K_1),K_2\right)
 = u_1 -  \Cop(u_1,u_2),
\end{align*}
and finally, since~$\Cf^2_+$ and~$\Pf^2_+$ are disjoint, 
\begin{align*}
\PP\left(\Cf^2_+\right) + \PP\left(\Pf^2_+\right)  & =\PP\left(S_T^2 > K_2 \right) + \PP\left(K_2 > S_T^2 \right)
= 1 - F_2\left(K_2\right) + F_2(K_2)
 = 1 - u_2 + u_2,
\end{align*}
\end{proof}

Armed with this lemma,
we can compute the conditional probability~\eqref{eq:muFormula} 
for the different European options we consider (Calls, Puts and Strangles):

\begin{proposition}
Let
$(u_1, \widetilde{u}_1, u_2, \widetilde{u}_2)
:=
(F_1(K_1), F_1(\widetilde{K}_1), F_2(K_2), F_2(\widetilde{K}_2))$,
with $K^i<\widetilde{K}^i$ for $i=1,2$.
Then the following equalities hold:
\begin{align*}
\mu\left(\Cf^1,\Cf^2\right)
 & = \frac{1 - u_1 - u_2 + \Cop(u_1,u_2)}{1 - u_2},\\
\mu\left(\Pf^1,\Pf^2\right)
 &  = \frac{\Cop(u_1,u_2)}{u_2},\\
\mu\left(\Sf^1,\Sf^2\right)
 &   =  \frac{1 + u_1 + u_2 - \widetilde{u}_1 - \widetilde{u}_2 + \Cop(u_1,u_2) + \Cop(\widetilde{u}_1, \widetilde{u}_2) - \Cop(u_1, \widetilde{u}_2) - \Cop(\widetilde{u}_2, u_1)}{1 - \widetilde{u}_2 + u_2},\\
\mu\left(\Cf^1,\Sf^2\right)
 &  = \frac{1- \widetilde{u}_1 - \widetilde{u}_2 + \Cop(\widetilde{u}_1,\widetilde{u}_2) + u_2 - \Cop(\widetilde{u}_1,u_2)}{1 - \widetilde{u}_2 + u_2},\\
\mu\left(\Pf^1,\Sf^2\right)
 & = \frac{u_1 - \Cop(u_1,\widetilde{u}_2) + \Cop(u_1,v_1)}{1 - \widetilde{u}_2 + u_2}.
\end{align*}
\end{proposition}

\begin{remark}
Note that if $u_1=u_2$, then~$\mu(\Cf^1,\Cf^2)$
and~$\mu(\Pf^1,\Pf^2)$
correspond exactly to the upper and lower tail dependence coefficients between~$S_T^1$ and~$S_T^2$
by Proposition~\ref{prop:TailDepCoef}.
\end{remark}

\begin{proof}
The first two equalities are straightforward since 
$$
\mu\left(\Cf^1,\Cf^2\right)
 = \PP\left(\Cf^1_+ \Big| \Cf^2_+ \right)
 = \frac{\PP\left(\Cf^1_+\cap\Cf^2_+\right) }{\PP\left(\Cf^2_+\right)}
$$
and 
$$
\mu\left(\Pf^1,\Pf^2\right)
 = \PP\left(\Pf^1_+ \Big| \Pf^2_+ \right)
 = \frac{\PP\left(\Pf^1_+\cap\Pf^2_+\right) }{\PP\left(\Pf^2_+\right)},
 $$
and Lemma~\ref{lem:AllProba} concludes.
Regarding the Strangles, we can write
$$
\mu(\Sf^1,\Sf^2)
= \PP\left(\Cft^1_+ \cup\Pf^1_+ \Big| \Cft^2_+ \cup\Pf^2_+\right)
= \frac{\PP\left(\{\Cft^1_+ \cup\Pf^1_+\}\cap\{ \Cft^2_+ \cup\Pf^2_+\}\right)}{\PP\left(\Cft^2_+ \cup\Pf^2_+\right)}.
$$
Since a Strangle does not have intersecting strike prices we can split the probabilities into disjoint events, 
i.e the probability of the Call part of the Strangle paying out at the same time as the Put is zero, 
so we can consider them separately:
\begin{align*}    
\mu(\Sf^1,\Sf^2)
 & = \frac{\PP\left(\Cft^1_+ \cap  \left\{\Cft^2_+ \cup \Pf^2_+\right\}\right) 
 + \PP\left(\Pf^1_+ \cap\left\{\Cft^2_+\cup \Pf^2_+\right\}\right)}{\PP\left(\Cft^2_+\cup\Pf^2_+\right)} \\
 & =  \frac{\PP\left(\Cft^1_+\cap\Cft^2_+\right)
+  \PP\left(\Cft^1_+\cap \Pf^2_+\right)
+ \PP\left(\Pf^1_+\cap\Cft^2_+\right)
+ \PP\left(\Pf^1_+\cap\Pf^2_+\right)}{\PP\left(\Cft^2_+\right) + \PP\left(\Pf^2_+\right)},
\end{align*}
and the representation in the proposition follows again by Lemma~\ref{lem:AllProba},
where analogous statements obviously hold
with $\widetilde{u}_1 = F_1(\widetilde{K}_1)$
and 
$\widetilde{u}_2 = F_2(\widetilde{K}_2)$
for~$\Cft^1_+$ and~$\Pft^1_+$
with the restriction 
$K_1< \widetilde{K}_1$
and $K_2<\widetilde{K}_2$.

Now, by definition,
\begin{align*}    
\mu(\Cf^1,\Sf^2)
&= \PP\left(\Cft^1_+ \Big|\left\{\Cft^2_+\cup\Pf^2_+\right\}\right)
= \frac{\PP\left(\Cft^1_+ \cap
\left\{\Cft^2_+\cup\Pf^2_+\right\}\right)}{\PP\left(\Cft^2_+\cup\Pf^2_+\right)}
= \frac{\PP\left(\left\{\Cft^1_+ \cap \Cft^2_+\right\} \cup \left\{\Cft^1_+ \cap \Pf^2_+\right\}\right)}{\PP\left(\Cft^2_+\right) + \PP\left(\Pf^2_+\right)}\\
&= \frac{1- \widetilde{u}_1 - \widetilde{u}_2 + \Cop(\widetilde{u}_1,\widetilde{u}_2) + u_2 - \Cop(\widetilde{u}_1,u_2)}{1 - \widetilde{u}_2 + u_2},
\end{align*}
with the denominator in~\eqref{eq:Ct2pP2p},
$\PP\left(\Cft^1_+\cap \Cft^2_+\right)$
in~\eqref{eq:Ct1pCt2p}
and
$\PP\left(\Cft^1_+\cap\Pf^2_+\right)$
in~\eqref{eq:Ct1pP2p}.
Finally, 
$$
\mu(\Pf^1,\Sf^2)= \PP\left(\Pf^1_+\Big| \left\{\Cft^2_+\cup
\Pf^2_+\right\}\right)
= \frac{\PP\left(\Pf^1_+ \cap \left\{\Cft^2_+\cup
\Pf^2_+\right\}\right)}{\PP\left(\Cft^2_+\cup
\Pf^2_+\right)}
 = \frac{\PP\left(\Pf^1_+ \cap \Cft^2_+\right) + \PP\left(\Pf^1_+ \cap \Pf^2_+\right) }{\PP\left(\Cft^2_+\cup
\Pf^2_+\right)},
$$
and the proposition follows by Lemma~\ref{lem:AllProba}.
\end{proof}

A simple application of the Fr\'echet-Hoeffding bounds~\eqref{eq:FrechetHoeffding} shows that
$\mu(\Cf^1,\Cf^2)$
and 
$\mu(\Pf^1,\Pf^2)$
are bounded between zero and~$1$,
and thus represent genuine probabilities.

\begin{remark}
We will only consider here options with the same maturities, but our framework naturally extends to options maturing at different times. 
\end{remark}

While the closed-form expressions are very attractive, these conditional probabilities are far from ideal. 
Suppose for example that the conditional probability that~$\Op^1$ pays out
given that~$\Op^2$ does equals~$70\%$, one might presume some sort of positive dependence between the two options. 
However if~$\Op^1$ pays out~$90\%$ of the time, 
this would indicate some negative correlation. 
On the other hand, it may be that the conditional probability that~$\Op^1$ pays out, given that~$\Op^2$ does is $20\%$, 
indicating negative dependence. 
However, if~$\Op^1$ normally only pays out $10\%$ of the time, this would mean that it pays out twice as much when~$\Op^2$ pays out. 

%%%%%%%%%%%%%%%%%%%%%%%%%%%%%%%%%%%%%%%%%%%%%%%%
%%%%%%%%%%%%%%%%%%%%%%%%%%%%%%%%%%%%%%%%%%%%%%%%
\subsection{Dependency matrices}
The above setup leads to a logical conditional dependency measure, which we define in the general and interesting framework of a portfolio consisting of~$n$ options.
We construct the dependency matrix $\LLa \in \RR^{n\times n}$, 
where, for $i,j=1,\ldots, n$, 
\begin{equation}\label{eq:Lambdaij}
\Lambda_{ij} := \frac{\mu(\Op^i,\Op^j)}
{\PP\left(\Op^i\right)}.
\end{equation}
Note that by construction, $\LLa$ is such that
\begin{equation}\label{eq:LambdaIneq}
\Lambda_{ii} = \frac{\mu(\Op^i,\Op^i)}
{\PP\left(\Op^i\right)}
 = \frac{1}{\PP\left(\Op^i\right)}
 \geq \frac{\mu(\Op^i,\Op^j)}
{\PP\left(\Op^i\right)}
 = \Lambda_{ij},
 \qquad\text{for }i,j=1,\ldots, n.
\end{equation}
When~$\Lambda_{ij}<1$, Option~$\Op^i$ pays out less often  when~$\Op^j$ does, 
indicating negative correlation. 
When $\Lambda_{ij}=1$, $\Op^i$ is not affected by the payout from~$\Op^j$, so that they are somehow  independent. 
Similarly a value strictly greater than~$1$ indicates positive correlation between the two options.
Clearly~$\LLa$ is symmetric with non-negative real entries, 
and therefore the spectral theorem~\cite{spectral} 
implies that there exist
a diagonal matrix~$\boldsymbol{D}$ of real eigenvalues of~$\LLa$ 
and an orthogonal matrix~$\Qm$ whose columns are the eigenvectors of~$\LLa$
such that $\LLa = \Qm\boldsymbol{D}\Qm^{\top}$.

%%%%%%%%%%%%%%%%%%%%%%%%%%%%%%%%%%%%%%%%
\subsubsection{Empirical motivation}

For clarity onwards,
we shall write
an~$x\%$ OTM option for an option being~$x\%$ out of the money;
for example, a~$5\%$ OTM Call on~$S$ has strike $K = 1.05 S_T$,
a~$10\%$ OTM Put on~$S$ has strike $K = 0.9 S_T$, and a $10\%$ either-side OTM Strangle has lower strike $K = 0.9 S_T$ and upper strike $\widetilde{K} = 1.1 S_T$.
Before analysing the conditional dependency matrix in detail, 
consider some empirical evidence. 
In Figure~\ref{fig:DepMat11},
the dependency matrix~$\LLa$ of~$5\%$, 
$30$-day-to-maturity European Calls
uses sample data from the last~$8$ years. Intuitively we expect to see clear dependence by sector, as stocks within sectors normally rise at similar times. 
There are three clear blocks with high dependence, centered around Technologies, Financials and Automobiles, indicating that the payouts between~$5\%$, $30$-day-to-maturity Call options are highly correlated intra-sector as expected. 
Some stocks such as Ford (F) have particularly high values on the diagonal, indicative of the poor recent performance of the stock and the low probability of returns greater than~$5\%$ during a $30$-day period.
Figure~\ref{fig:DepMat12} scales it up
and includes~$5\%$, 30-day-to maturity European Puts and~$2\%$, $30$-day-to-maturity European Strangles. 
There is high dependence between diagonal blocks as expected. 
Additionally, the Strangle has cross dependence with both Puts and Calls, importantly almost always on the same stocks. 
The Call-Put dependence blocks also show the expected dependence structure, 
close to zero for all coefficients indicating negative correlation between the option payouts, 
reflecting that collectively stocks are likely to rise or fall at similar times in bull or bear markets.
\begin{figure}[h!]
\centering
\includegraphics[scale=0.4]{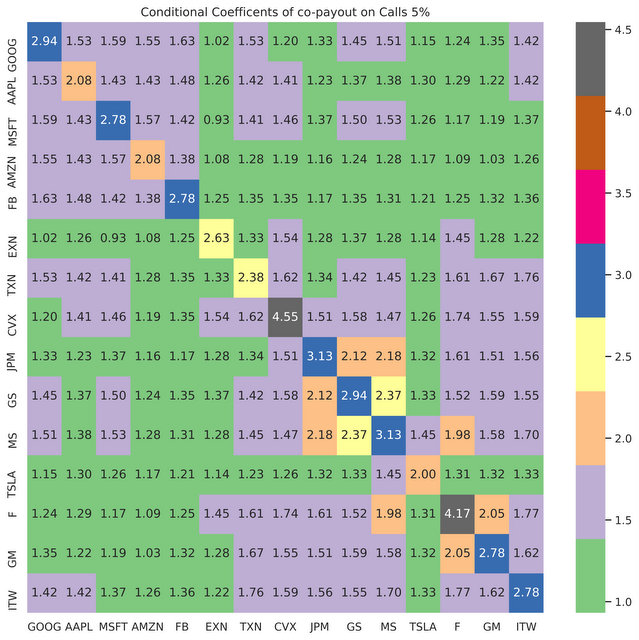}
\caption{Dependency matrix of $5\%$ OTM Call options.
}
\label{fig:DepMat11}
\end{figure}
\begin{figure}[h!]
\centering
\includegraphics[scale=0.45]{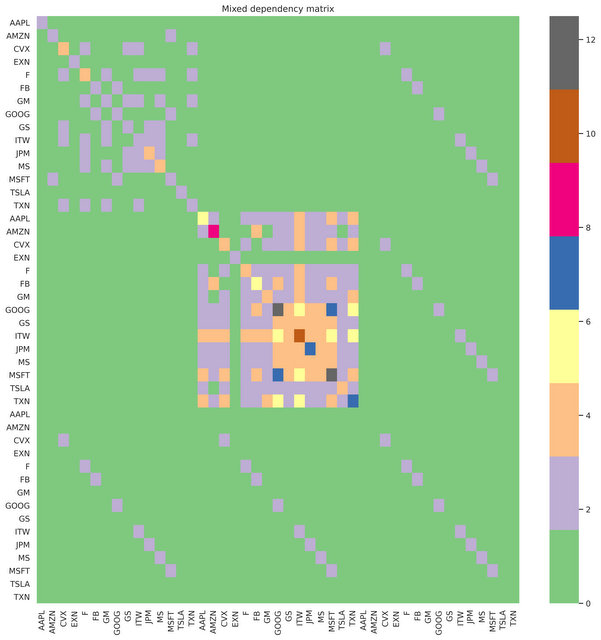}
\caption{Dependency matrix of $5\%$ and $2\%$ OTM Call, Put and Strangle options.}
\label{fig:DepMat12}
\end{figure}

%%%%%%%%%%%%%%%%%%%%%%%%%%%%%%%%%%%%%%%%%%%%%%%%%%%
\subsubsection{Properties of $\LLa$}\label{sec:ModCholesky}

\subsubsection*{Positive semi-definiteness}
While the dependency matrix~$\LLa$ is symmetric,
it may not be positive semi-definite
($\xx^{\top} \LLa \xx \geq 0$ 
for all $\xx \in \RR^{n}$), 
as some eigenvalues may be negative.
The Perron–Frobenius theorem~\cite{Perron1907} however ensures that the largest eigenvalue for any irreducible non-negative matrix is greater than zero and all others are bounded by its absolute value.  
Moreover, from~\eqref{eq:LambdaIneq},
Gershgorin's circle theorem~\cite{circle} 
provides bounds for the eigenvalues~$\{\lambda_i\}_{i=1,\ldots,n}$
of the form
$|\lambda_i -\Lambda_{ii}|\leq \sum_{j \neq i}\left|\Lambda_{i j}\right|$
for each $i=1,\ldots,n$.
In the case where~$\LLa$ 
is diagonally dominant, namely
$|\Lambda_{ii}|\geq \sum_{j \ne i}|\Lambda_{ij}|$
for each $i=1,\ldots,n$, 
then it has positive eigenvalue, 
and hence is positive semi-definite.
While we have found this to hold empirically,
there is no guarantee in principle that this will always be true,
and we show how to adjust~$\LLa$ 
to account for this possibility.

%%%%%%%%%%%%%%%%%%%%%%%%%%%%%%%%%%%%%%%%%%%%%%
\subsubsection*{Modified Cholesky}
If~$\LLa$ is not positive semi-definite, 
quadratic programming with even a single negative eigenvalue is known to be NP-hard~\cite{NP},
and therefore rather cumbersome from a numerical point of view.
We use instead a modified Cholesky algorithm to modify~$\LLa$ slightly to make it positive semi-definite.
We consider the set of diagonal matrices $\mathcal{D}$ and 
following Higham and Cheng~\cite{MCA_hig},
for  a small constant $\delta \geq 0$,
introduce
$$
\eta(\LLa, \delta)_F
:= \min_{\Eb\in\mathcal{D}}\Big\{\|\Eb\|: \lambda_{\min }(\LLa+\Eb) \geq \delta\Big\},
$$
for any matrix norm~$\|\cdot\|$.
In the particular case of the Frobenius norm, it can be shown that~\cite[Theorem 3.1]{MCA_hig}
$$
\eta(\LLa, \delta)_{F}
=\left(\sum_{\lambda_{i}<\delta}\left(\delta-\lambda_{i}\right)^{2}\right)^{1/2},
$$
and the optimal perturbation matrix~$\Eb$ takes the form
\begin{equation}
\Eb=\Qm\, \mathrm{Diag}(\taub)\,\Qm^{\top}, 
\qquad \text{with}\qquad
\tau_{i} := 
\left\{
\begin{array}{ll}
0, & \text{if }\lambda_{i} \geq \delta, \\ 
\delta-\lambda_{i}, & \text{if }\lambda_{i}<\delta.
\end{array}
\right.
\end{equation}

\begin{remark}
Note that setting $\delta>0$ from the beginning 
will return a positive definite modified matrix,
which then admits an inverse.
\end{remark}
%%%%%%%%%%%%%%%%%%%%%%%%%%%%%%%%%%%%%%%%%%%%%%%%%%%
%%%%%%%%%%%%%%%%%%%%%%%%%%%%%%%%%%%%%%%%%%%%%%%%%%%
\section{Portfolio Optimisation}\label{sec:PtfOptim}

We now tackle the core of the paper, showing how the dependency matrix becomes a key tool for options portfolio optimisation.
For notational clarity, we shall from now on assume that the dependency matrix~$\LLa$
is positive semi-definite,
either by construction or after applying the modified Cholesky algorithm from Section~\ref{sec:ModCholesky}.

%%%%%%%%%%%%%%%%%%%%%%%%%%%%%%%%%%%%%%%%%%%
\subsection{Investment procedure}
We consider a trader who invests one unit of capital in a portfolio of options every~$d$ days. 
She can buy from a total of~$n$ different options, 
written on $k\leq n$ underlying stocks with a maturity of~$d$ trading days, which she then holds to maturity. The options traded in this market include various combinations of Calls, Puts, Strangles and Straddles. If at maturity the options have a positive payout,
the trader cashes them in, 
otherwise she loses all the capital invested in that option. 
The returns~$R_i$ on an option~$\Op^i$ can be characterised as
$$
R_i := \frac{\text{Payout}(\Op^i)  - \text{Price}(\Op^i)}{\text{Price}(\Op^i)},
$$
where $\text{Price}(\Op^i)$ denotes the price of~$\Op^i$ at inception of the contract,
and the weight of  each option~$\Op^i$ in the portfolio is defined as
$$
w_i := \frac{\text{Number of }\Op^i \text{ bought}}{\text{Price}(\Op^i)}.
$$
This is of course an ideal setup,
as in practice options contracts traded on S$\&$P 500 stocks are traded in multiples of~$100$,
but this is flexible enough to be updated for practical scenarii.
The returns~$R^{\pi}$ of the portfolio over the~$d$-day period is therefore
$$
R^{\pi} := \ww^{\top}\Rb = \sum_{i=1}^{n} w_i R_i,
$$
where $\ww = (w_1, \ldots, w_n)$
and $\Rb = (R_1,\ldots, R_n)$. 
Note that this is an `all-or-nothing strategy'
and the trader can potentially lose all her initial capital if none of the options pays out at maturity. 
This inherent risk in investing in pure option portfolios
(especially in deep out-of-the-money options)
differs greatly from investing in stocks where
it is highly unlikely that all the stocks become worthless.

%%%%%%%%%%%%%%%%%%%%%%%%%%%%%%%%%%%%%%%%%%%%%%%%%%
\begin{remark}
We only consider portfolios with naked options rather than with options and underlying assets.
While this may seem like a risky portfolio (as discussed later), one could construct the corresponding delta-hedged portfolio consisting of assets only,
where the amount of each asset is the sum of all the options on this asset, weighted by their deltas.
\end{remark}

%%%%%%%%%%%%%%%%%%%%%%%%%%%%%%%%%%%%%%%%%%%%%%%%%%%%%
\subsection{Objective function}

The goal is to maximise the portfolio returns 
for a given level of risk. 
As discussed in the introduction, 
Markowitz' portfolio theory uses the returns covariance as a measure of the risk, 
which is a poor metric here because of the highly asymmetric returns distributions. 
We consider instead the minimisation problem
$$
\min_{\{\ww: \ww^\top\UnitVec = 1\}}\Big\{ -\ww^\top \EE[\Rb] + \alpha \ww^\top\LLa\ww\Big\},
$$
for some risk aversion parameter~$\alpha\geq 0$.
The intuition here is that the modified dependency matrix~$\LLa$ has high diagonal terms when option payout probabilities are historically low, 
which means that weights with low probability of positive payout are penalised. 
While off-diagonal terms are high if correlation between positive payouts on options are high, 
thus penalising weights with high correlation of payouts, hopefully leading to a more diverse portfolio.
Since the matrix~$\LLa$ is positive semi-definite, the objective function is convex and can be be solved explicitly using Lagrange multipliers. 
The Lagrangian $\mathcal{L}$ reads
$$
\mathcal{L}(\ww,\gamma)
:= -\ww^\top \EE[\Rb] + \alpha \ww^\top\LLa\ww + \gamma \left(\ww^\top \UnitVec - 1\right).
$$
Equating the gradients~$\nabla_{\ww}$ and~$\partial_{\gamma}$ to zero yields the system
\begin{equation*}
\left\{
\begin{array}{rl}
\ZeroVec &= -\EE[\Rb] + 2\alpha\LLa\ww + \gamma \UnitVec,\\
0 &= \ww^\top \UnitVec,
\end{array}
\right.
\end{equation*}
which can be solved explicitly as
\begin{equation}\label{eq:SolutionOptim}
\gamma = \frac{2\alpha - \UnitVec^\top \LLa^{-1}\EE[\Rb] }{\UnitVec^\top\LLa^{-1}\UnitVec}
\qquad\text{and}\qquad
\ww = \frac{\LLa^{-1}}{2\alpha} \left(\EE[\Rb] - \UnitVec \frac{2\alpha -  \UnitVec^\top \LLa^{-1}\EE[\Rb]}{\UnitVec^\top\LLa^{-1}\UnitVec} \right).
\end{equation}

\begin{remark}
The solution~\eqref{eq:SolutionOptim} 
assumes that the matrix~$\LLa$ is invertible.
The modified Cholesky Algorithm can ensure invertibility setting $\delta>0$, 
or otherwise we can rely instead on the Moore-Penrose inverse~\cite{moore, pen}.
\end{remark}

%%%%%%%%%%%%%%%%%%%%%%%%%%%%%%%%%%%%%%%%%%%%%%%%%%%%%%%%
\subsection{Optimisation with box constraints}

The framework developed above has several limitations:
the quadratic nature of the risk aversion
does not distinguish between selling and buying, and hence is not penalising selling options with high probability of payouts. 
To palliate this issue, we impose short-selling constraints, thus reducing the risk, as the trader is no longer exposed to the unlimited downside created from selling naked options.  
The trader (because of regulations for example) 
may also wish to impose diversification constraints, either imposing some non-zero weights or limiting individual amounts;
she may also have target proportions of the portfolio in certain categories of assets.
The new objective function thus takes the form
\begin{equation}\label{eq:minModif}
\min_{w_i}\left\{\ww^\top \EE[\Rb] - \alpha\ww^\top\LLa\ww:
\ww^T\UnitVec = 1,
\sum^{m}_{i = 1} w_i = x,
l_{i} \leq w_i \leq u_{i}
\text{ for }i=1, \ldots, n\right\}.
\end{equation}

Since~$\LLa$ is positive semi-definite, 
the problem~\eqref{eq:minModif} is convex quadratic with box constraints. 
Several classes of algorithms exist
to solve such problems, such as Sequential Quadratic Programming~\cite{SQP} 
or Augmented Lagrangian Methods~\cite{birgin2014practical}. 
We choose the former, a sequential linear-quadratic programming (SLSQP)~\cite{Schittkowski1982}, 
which is fast and handles well those inequality constraints. 
We refer the interested reader to~\cite{boggs}
for an excellent overview.

%%%%%%%%%%%%%%%%%%%%%%%%%%%%%%%%%%%%%%%%%%%%%%%%%%%
%%%%%%%%%%%%%%%%%%%%%%%%%%%%%%%%%%%%%%%%%%%%%%%%%%%
\section{Results}\label{sec:Results}
We now perform out-of-sample backtests over the last nine years using a variety of OTM options traded on $15$ S$\&$P 500 stocks, 
assuming all options are held to maturity. 
Weights are selected by optimising~\eqref{eq:minModif} and we also consider different levels of risk aversion~$\alpha$ to assess their impact on portfolio  performance. 
Performance is measured using the Sharpe ratio, the CRRA utility function,
as well as with moments of the returns distribution. As a comparison, the performance of these portfolios will be benchmarked against an equally weighted portfolio.

%%%%%%%%%%%%%%%%%%%%%%%%%%%%%%%%%%%%%%%%%%%%%%%%%%%%%
\subsection{Portfolio construction}
We consider three trading scenarios: 
\begin{itemize}
\item $5\%$ OTM Call options only;
\item $10\%$ OTM Call options;
\item a~$5\%$ OTM Call, a~$5\%$ OTM Put and a~$10\%$ either-side Strangle.
\end{itemize}

All options are bought with one month to maturity 
and traded on all $15$~S$\&$P 500 stocks, 
selected from a diverse range of industries. 
Note that the price of the $10\%$ either side Strangle can be calculated as the sum of $10\%$ OTM Call and Puts. 
This choice of options and stocks 
is dictated by liquidity considerations,
so that we likely avoid absence of prices for some options on some days. 
We assume that options can be bought in any quantity, even fractional. 
While this is not completely true,
as options are usually bought in groups,
this should not deter the reader from the main message.

We build our trading strategy over~$N$ one-month periods 
$[t_i, t_{i+1}]$
for $i=0, \ldots, N-1$.
At the beginning~$t_i$ of each period
we compute the  dependency matrix~$\LLa_{t_i}$ 
(and its modified Cholesky version)
with the best of three Archimedean copulas (Frank, Gumbel and Clayton)
as well as the expected returns vector~$\EE_{t_i}[\Rb]$,
using historic rolling sampling over the past eight years.
At time~$t_i$ we optimise~\eqref{eq:minModif} with $\LLa_{t_i}$ and $\EE_{t_i}[\Rb]$ 
for a selected value of the risk parameter~$\alpha$ to obtain the optimal weights~$\ww_{t_i}$. 
At time~$t_{i+1}$, we then compute 
the returns of the portfolio $R^{\pi}_{t_i,t_{i+1}}$ over the elapsed period $[t_i, t_{i+1}]$. 
The total average returns over the~$N$ periods is therefore
$$
\overline{R}^{\pi} =
\frac{1}{N}\sum_{i=0}^{N-1}R^{\pi}_{t_i,t_{i+1}},
$$
where we assume that we invest the same amount of capital at the start of each period as a normalisation.
This analysis is inherently out of sample, only using data from before $t_i$ at each timestep to create the dependency matrix $\LLa_{t_i}$. Moreover as the window is rolling, we are able to encapsulate changes in the dynamics between various different stocks.

We evaluate the performance of the portfolio using the Sharpe ratio~\cite{sharpe},  
$\mathrm{SR}:=\frac{1}{\sigma}(\overline{R}^{\pi}-R_{f})$,
a standard metric in portfolio management.
Here, $\sigma$ represents the standard deviation of returns 
($\sigma^2 = \sum_{i=0}^{N-1}\sigma_{t_i,t_{i+1}}^2$)
and where~$R_{f}$ is the risk-free rate. In particular, the Sharpe ratio is annualised using monthly returns.  We also compute the CRRA utility function given by 
$u(c) = \frac{c^{1-\eta}-1}{1-\eta}$ with $\eta=1.5$

As we have a no short-selling constraint, 
we are only buying OTM options and are therefore long volatility. We stand to make a lot of money when the market crashes or rises, but in stable markets will steadily lose money. For example, holding a long option portfolio in Spring 2020 would have been hugely profitable as the market crashed and then rebounded. Therefore we expect to see returns with high positive skew and excess kurtosis, with the majority of returns being negative.

%%%%%%%%%%%%%%%%%%%%%%%%%%%%%%%%%%%%%%%%%%%%%%%%%%%%%
\subsection{Data}
We have at our disposal options data from February 2005 to February 2022. 
We select OTM Call and Put options with a maturity of $20$ business days. 
We test the strategy from February 2013 to February 2022, since the initial rolling window requires samples from February~2005.
%We test the strategy over the period May~2017 to May~2021, requiring samples from May~2009. 
This period contains a wide variety of market conditions, including highly volatile periods such as the Covid-related market crash in Spring~2020, 
where the S$\&$P 500 fell over~$9.5\%$ on March 12 alone.
We only trade options with monthly expirations, usually expiring on the third Friday of each month~\cite{Faias2017OptimalOP}.
The only data required for the backtest are the one-month maturity option smiles
and their respective stock prices at maturity.
The vast majority of S$\&$P 500 options on individual stocks are American and not European.
While their prices are slightly higher than 
their European counterparts, 
we treat them as holding the contracts to expiry. 
We choose to trade specific S$\&$P 500 stocks that were listed continuously throughout the sampling and testing periods,
excluding for example recently listed stocks (Tesla, Facebook).
The~$15$ liquidly traded stocks are GOOG, AAPL, MSFT, AMZN, OXY, XOM, HAL, CVX, JPM, GS, MS, JNJ, ABT, F, ITW,
distributed between Technology, Financials, Automobiles, Pharmaceuticals and Energy. 
As we reinvest monthly in new options we should give consideration to transaction costs. 
We have assumed transaction costs on these options to be negligible as these are all liquid options; 
in practice, they are small enough not to affect the shape of our results.

%%%%%%%%%%%%%%%%%%%%%%%%%%%%%%%%%%%%%%%%%%%%%%%%%%%%%
\subsection{Call option portfolios}
We first concentrate on portfolios consisting of $5\%$ and $10\%$ OTM Calls.  
We expect to reduce the risk as we increase the risk aversion parameter~$\alpha$, 
with high dependence being more penalised and weights more evenly distributed between sectors and stocks.
To test this we run backtests on $\alpha \in\{0, 10\}$ 
and observe the behaviour of the optimal portfolio weights over time.
We first study the effect of the dependency matrix on the optimal distribution of the weights. 
When $\alpha = 0$, the optimisation problem~\eqref{eq:minModif} is merely maximising expected returns, without taking into account dependencies between the options. 
In this case the portfolio composition is rather static, as we observe in most cases only two symbols rebalancing the same month in Figure~\ref{fig:CallWeightsAlpha0}. 
In particular, the $10\%$ Call optimal weights are almost constant over the whole backtest horizon. 
The historically top performing stocks are given the maximum amount of capital while the rest are given minimal allocation. 
As we increase~$\alpha$, the optimisation penalises dependence and low payout probabilities such that the portfolio diversifies into presumably less dependent Call options.
This is well captured in Figure~\ref{fig:CallWeightsAlpha10}
where $\alpha=10$ creates a more dynamic portfolio.
When computing dependency coefficients, there are few instances where the denominator~$1-u_j$ is null and the coefficients diverge $\Lambda_{ij} = +\infty$. This occurs when a stock has never seen a~$10\%$ or more rise in value in a single month during the $8$-year rolling sampling window, so that such return is close to the 100th percentile of the distribution. One solution is to set a maximum for the percentiles~$u_i$ and~$u_j$, or the dependency coefficient themselves. 
We leave refinements of this issue to further investigations.

\begin{figure}[h!]
\centering
\includegraphics[scale=0.22]{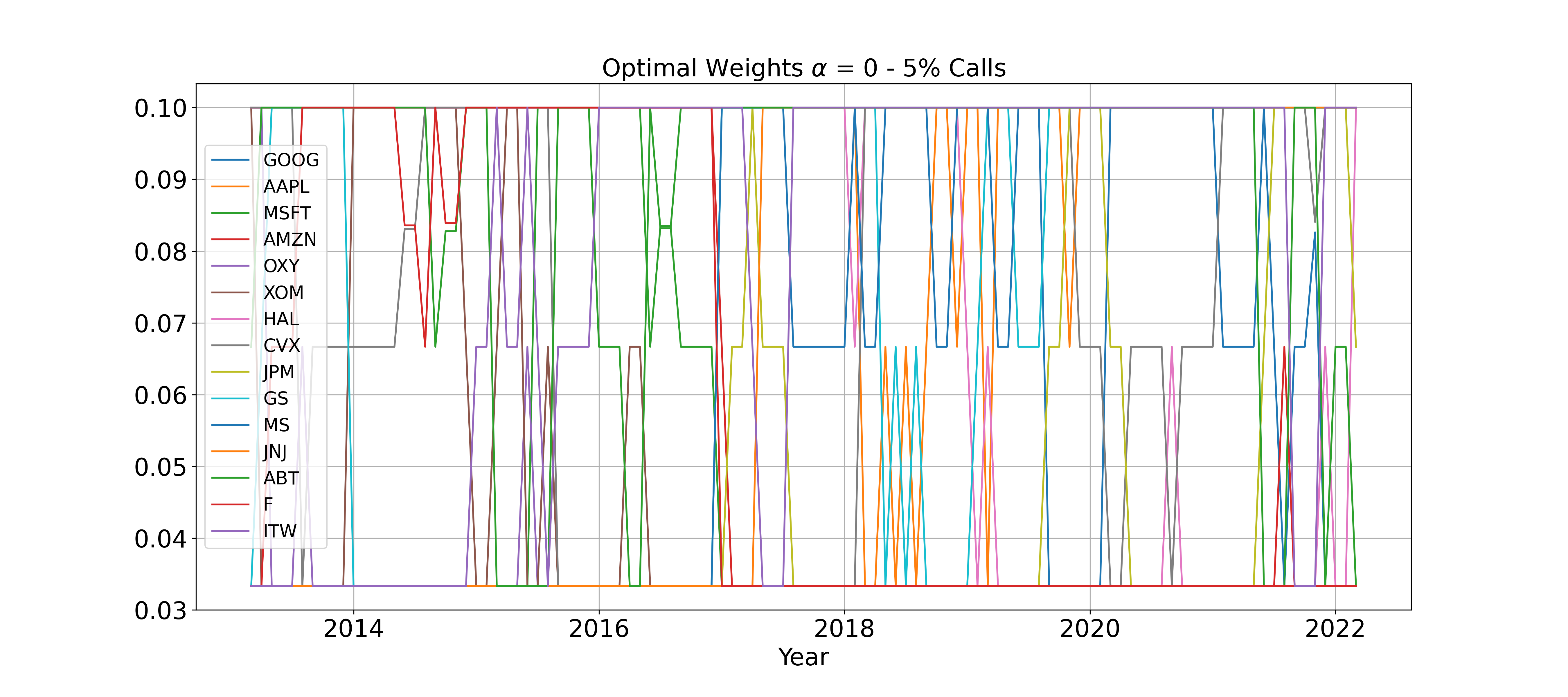}
\includegraphics[scale=0.22]{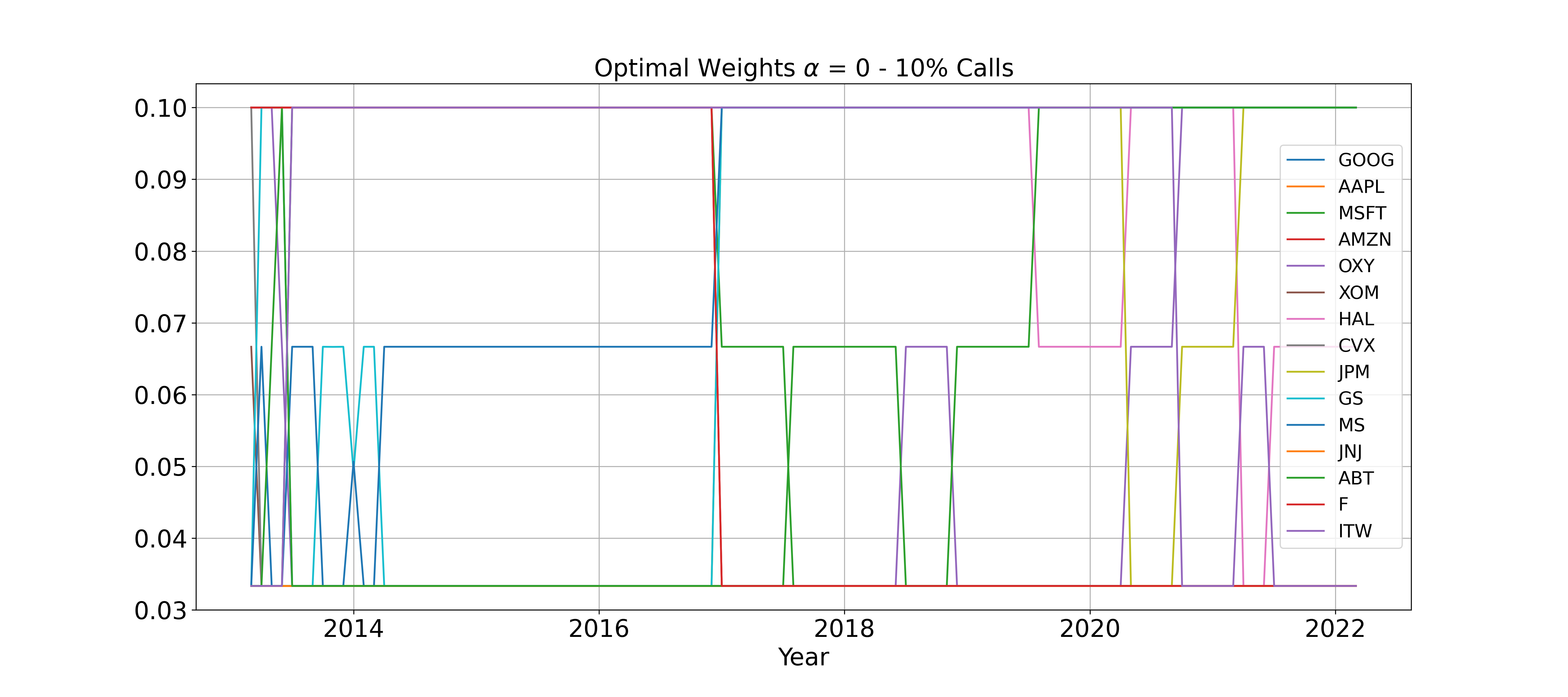}
\caption{$5\%$ and $10\%$ Calls weights allocation with $\alpha = 0$.}
\label{fig:CallWeightsAlpha0}
\end{figure}

\begin{figure}[h!]
\centering
\includegraphics[scale=0.22]{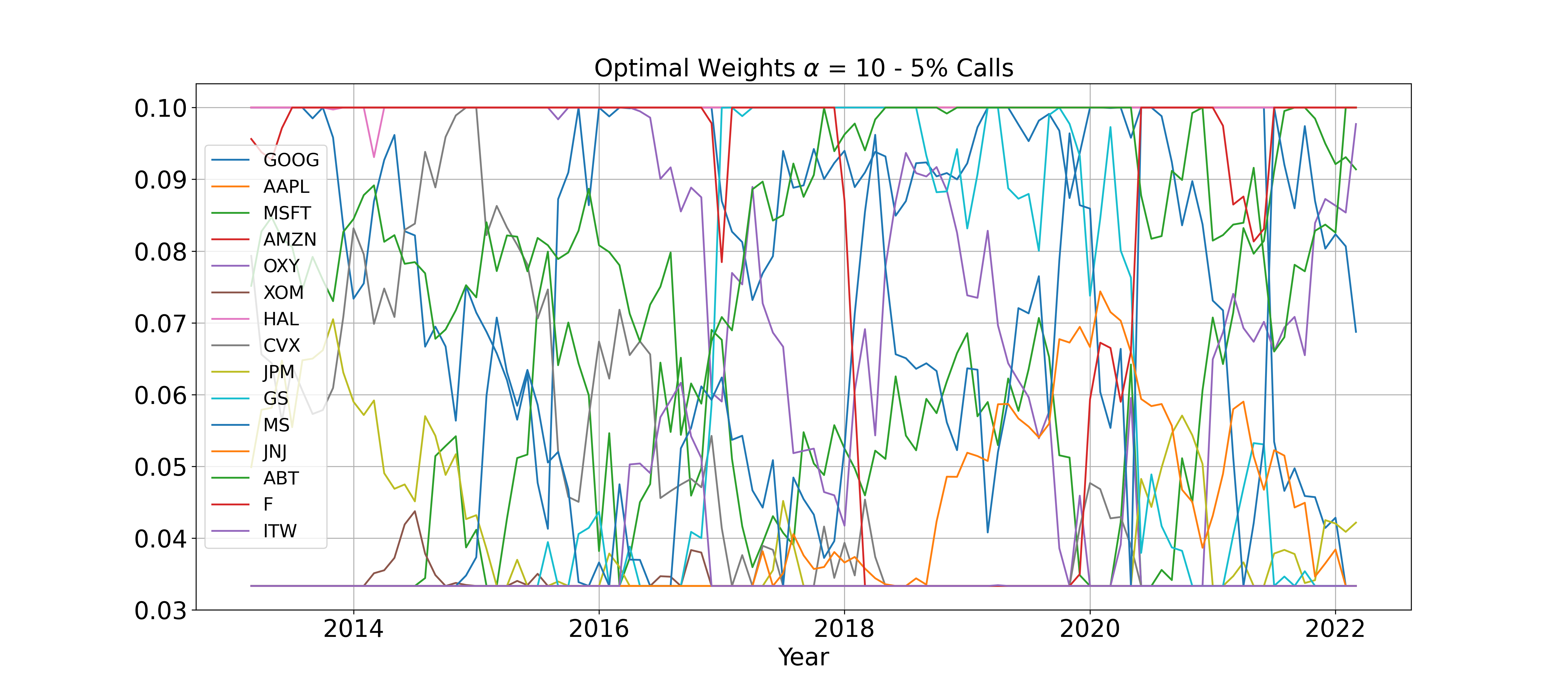}
\includegraphics[scale=0.22]{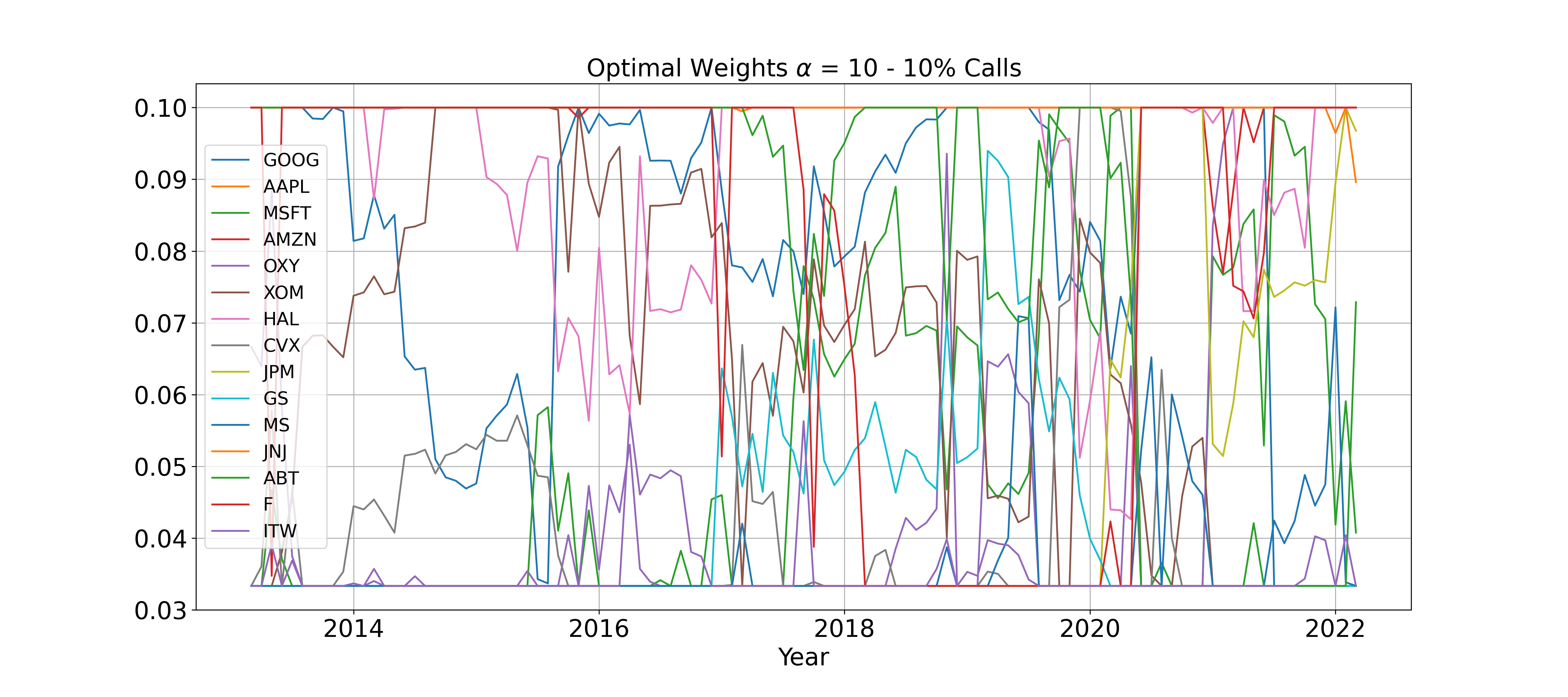}
\caption{$5\%$ and $10\%$ Calls weights allocation with $\alpha = 10$.}
\label{fig:CallWeightsAlpha10}
\end{figure}

We now look at the portfolios' performances for different values of $\alpha \in\{0, 2, 5, 10\}$. 
We compare them with the Equally Weighted Portfolio (EWP) of Calls, as well as the EWP of underlying stocks (EWPu).
Since we only consider one-month OTM Call options, their likelihood of paying out is small. We then expect to see, even with a very well selected portfolio, many months where returns are equal to~$-100\%$. Conversely deep OTM options are very cheap 
and thus promise large returns if stocks rise significantly over~$10\%$. 
This would yield returns distribution with positive skew and high kurtosis.
Figure~\ref{fig:Call OTM PnL Hist Raw} shows the monthly P\&L distribution for $5\%$ and $10\%$ Calls portfolios. The distribution is indeed right skewed with a Dirac mass at~$-100\%$, which represents the months were all Call options expired worthless. In particular, we observe that this mass is much lower for optimised portfolios compared to the EWP for $10\%$ Calls. It is not the case for $5\%$ Calls, even though the density between $-50\%$ and $0\%$ P\&L is much lower for optimised portfolios compared to the EWP. All distributions are fat tailed with a few
points at more than $1500\%$ P\&L. Given how deep out-of-the-money the considered Call options are, the strategy is indeed highly risky with potentially extreme returns.

\begin{figure}[h!]
\centering
\includegraphics[scale=0.22]{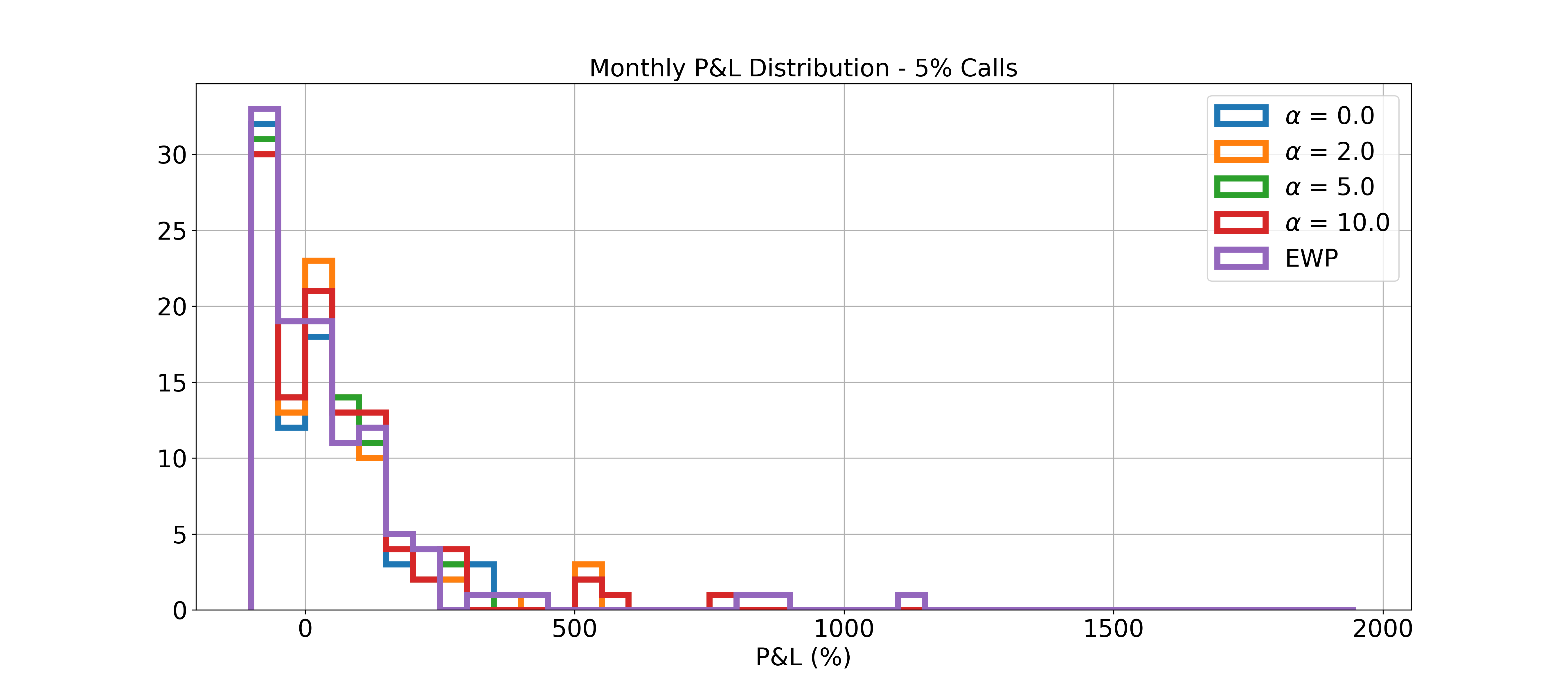}
\includegraphics[scale=0.22]{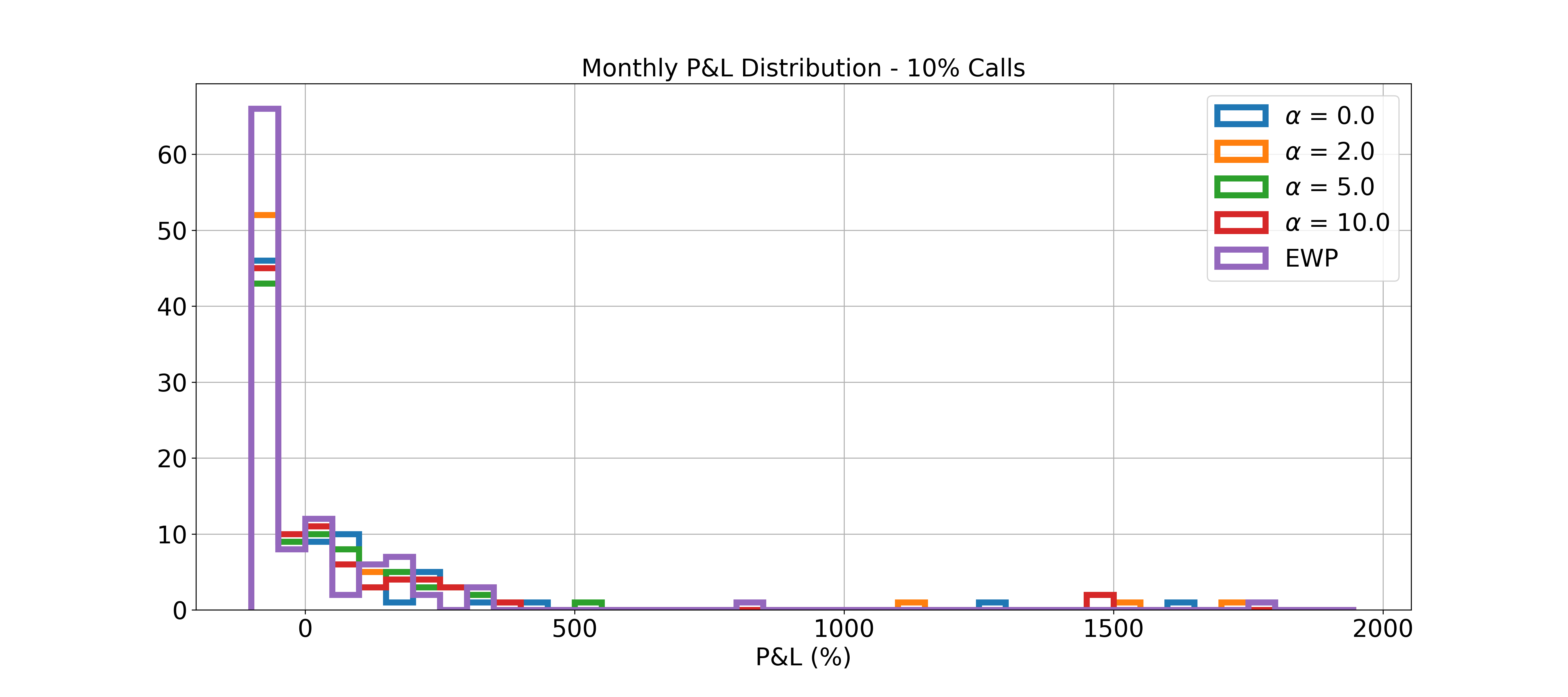}
\caption{Monthly P\&L distribution; $5\%$ and $10\%$ Calls portfolio, EWP Calls portfolio and EWP Underlying Equity portfolio.}
\label{fig:Call OTM PnL Hist Raw}
\end{figure}

Accounting for the leveraged nature of OTM options, we choose to invest a 1/12th of our total capital to buy our options portfolio every month.
Figure~\ref{fig:Call OTM PnL Cumul Raw} compares the cumulative P\&L of portfolios over the backtest horizon. For $5\%$ Calls, the EWP seems to outperform other portfolios starting from early~2017. The optimised portfolios seem to have the same performance, except for late 2021 where the portfolio with risk aversion $\alpha=0$ falls behind. The trend is rather flat between~2013 to~2016 and starts going steadily upwards from~2017. 
The overall strategy seems profitable and beats the EWPu.
Regarding $10\%$ Calls, we observe three significant jumps in early~2013, mid~2015 and late~2016. 
The portfolios' values skyrocket for these single expirations. Otherwise, the trend seems to go consistently downwards. 
The portfolios with $\alpha=5, 10$ and the EWP barely break even. 
For $\alpha=0, 2$, the final P\&L is better but is still lower than the EWPu.

\begin{figure}[h!]
\centering
\includegraphics[scale=0.22]{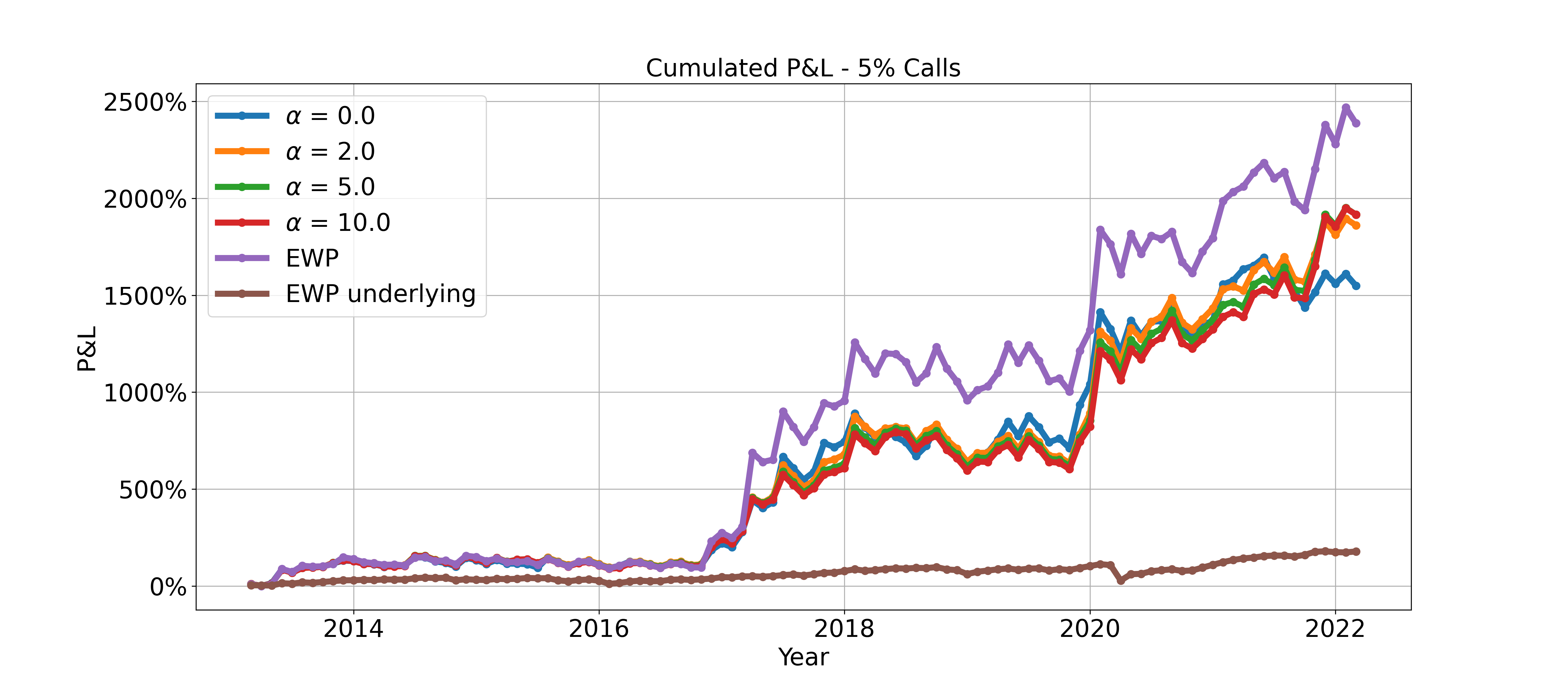}
\includegraphics[scale=0.22]{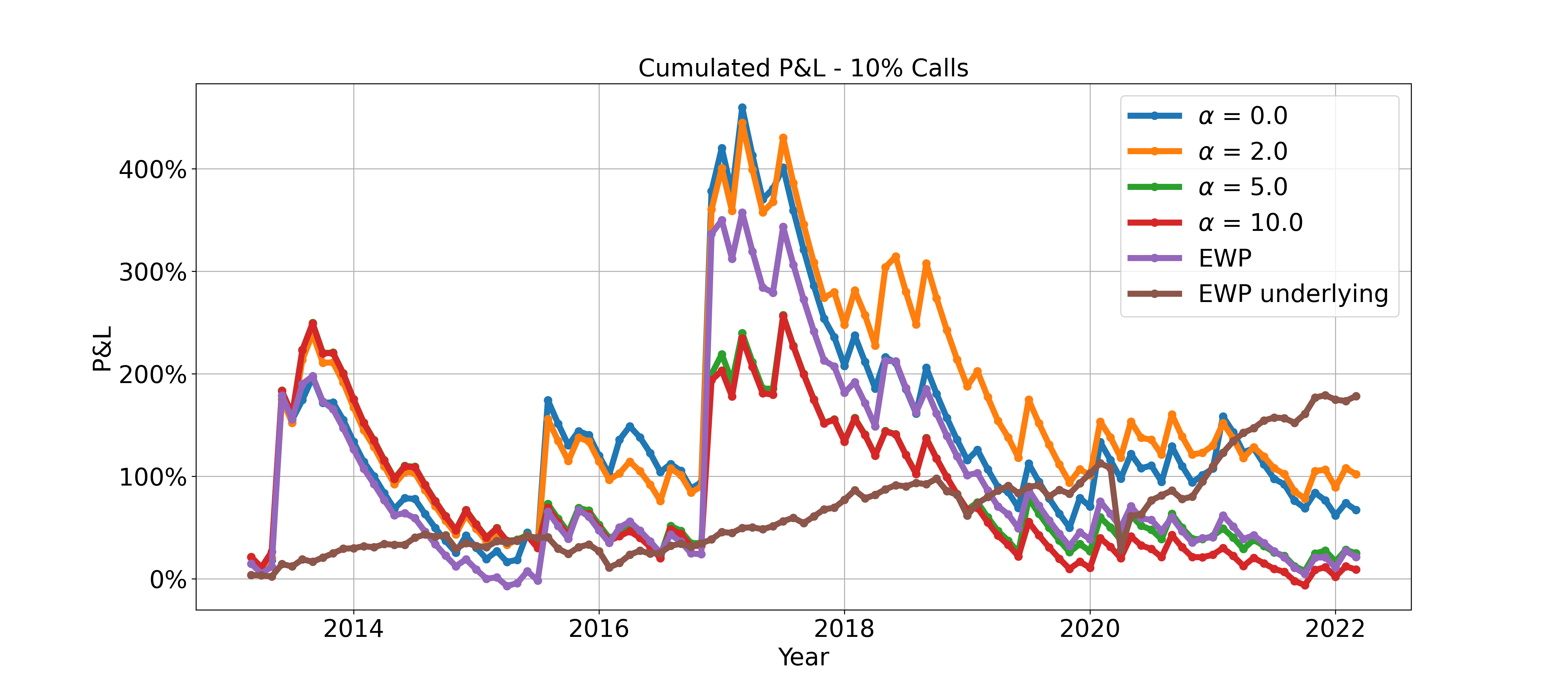}-
\caption{Monthly cumulative P\&L; $5\%$ and $10\%$ Calls portfolio, EWP Calls portfolio and EWP Underlying Equity portfolio.}
\label{fig:Call OTM PnL Cumul Raw}
\end{figure}

Tables~\ref{tab:Stats Calls 5 Raw} and~\ref{tab:Stats Calls 10 Raw} summarise some statistics related to the performance of the~$5\%$ and the~$10\%$ Calls portfolios. 
In particular, we look at the Sharpe Ratios, as well as the CRRA utility function.
As expected, all portfolios have high kurtosis and skew, with low Sharpe ratios. The EWP portfolios sensibly have the highest values of skewness and kurtosis in both the~$5\%$ and~$10\%$ cases, compared to the optimised portfolios. In the case of~$5\%$ Calls, the EWP's P\&L seems to outperform. 
The Sharpe ratios and CRRA values are similar across all portfolios. As for the~$10\%$ Calls, the portfolio with risk aversion $\alpha=2$ has the best P\&L and CRRA value. 
All portfolios have similar Sharpe Ratios, but EWP has the lowest CRRA value.

\begin{table}[h!]
\begin{tabular}{|c|c|c|c|c|c|} 
\hline
Statistics  & $\alpha = 0$ &$\alpha = 2$&$\alpha = 5$&$\alpha = 10$& EWP\\
 \hline
Total return ($\%$)            & 1550 & 1860  & 1920 & 1920 & 2390 \\
 Average monthly return ($\%$) &  39.4 & 40.9  & 41.2 & 41.2 & 47.2   \\
 Skewness                      & 1.94  & 2.11  & 2.16  & 2.18  & 3.26  \\
 Kurtosis                      &  4.83 & 6.09  & 6.47 & 6.61 & 13.5   \\
 Annualised Sharpe Ratio                  &  0.90 & 0.95  & 0.96 & 0.96 & 0.86  \\
 CRRA (1e-2)                   &  2.23 & 2.43  & 2.45 & 2.46 & 2.54   \\
 \hline
\end{tabular}
\caption{$5\%$ Call portfolio Statistics summary.}
\label{tab:Stats Calls 5 Raw}
\end{table}

\begin{table}[h!]
\begin{tabular}{|c|c|c|c|c|c|} 
\hline
Statistics  & $\alpha = 0$ &$\alpha = 2$&$\alpha = 5$&$\alpha = 10$& EWP\\
 \hline
Total return ($\%$)            & 67 & 102 & 25 & 9 & 21\\
 Average monthly return ($\%$) & 26 & 27 & 18 & 16 & 29  \\
 Skewness                      & 4.55 & 4.50 & 4.46 & 4.56 & 6.39  \\
 Kurtosis                      & 22.6 & 22.7 & 23.9 & 25.0 & 45.8  \\
 Annualised Sharpe Ratio                  & 0.32 & 0.34 & 0.26 & 0.24 & 0.28  \\
 CRRA (1e-4)                   & -15 & -5 & -30 & -41 & -53  \\
 \hline
\end{tabular}
\caption{$10\%$ Call portfolio Statistics summary.}
\label{tab:Stats Calls 10 Raw}
\end{table}

As noted earlier, returns are most often negative and the P\&L generally trends downwards,
except for some months where large returns push the whole portfolio returns up. Especially, the $10\%$ Calls portfolio relies on a huge swings to break even. 
Figure~\ref{fig:CallOTMSectorSelect} shows the EWP's P\&L breakdown by sector for selected expiration dates where returns are extremely high. These are 2016-11, 2017-03, 2017-06 for $5\%$ Calls and 2013-05, 2013-07, 2013-08, 2015-07, 2016-11 for $10\%$ Calls. We observe that for a few periods, the P\&L comes from a unique sector, namely Financials for 2016-11 with a tremendous $3000\%$ return on allocated capital, Health care for 2017-03 and 2017-06, and Technology for 2015-07. Even within the same sector, symbols appear to have uneven contribution. A slightly different choice of weights on symbols significantly impacts the overall portfolio performance. For this reason, we decide to discard these periods as they lack relevance for our analysis. The dates 2013-05, 2013-07, 2013-08 are kept since the P\&L contribution comes from multiple sectors and is fairly balanced.

\begin{figure}[h!]
\centering
\includegraphics[scale=0.22]{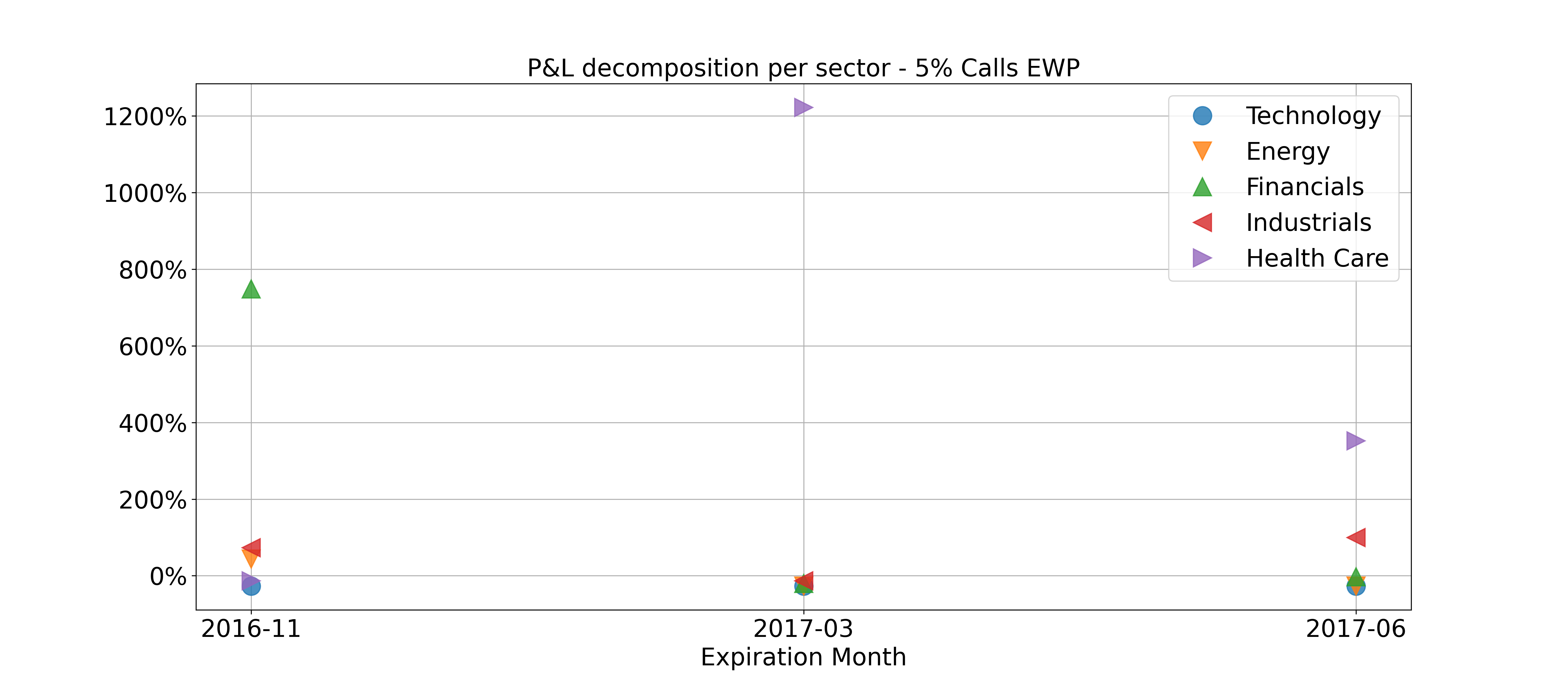}
\includegraphics[scale=0.22]{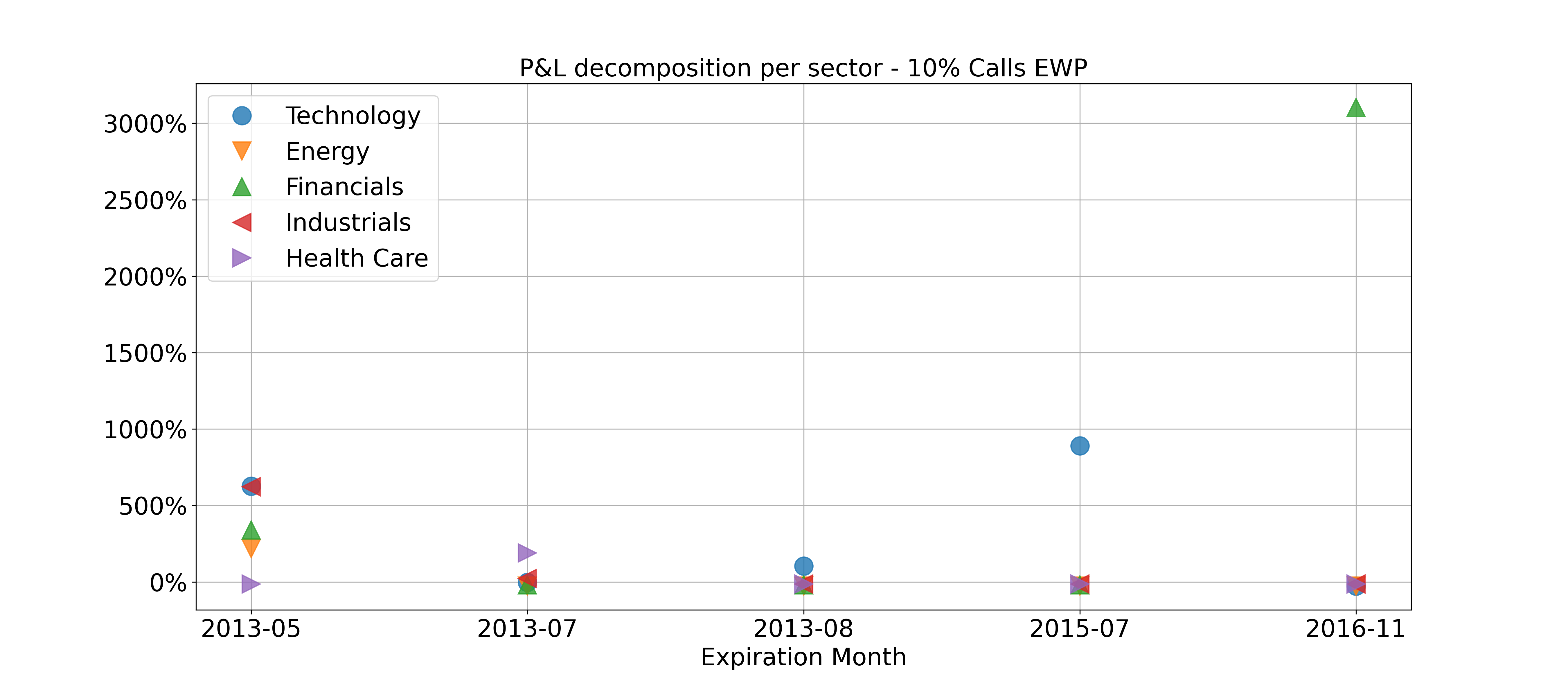}
\caption{$5\%$ and $10\%$ Call EWP's P\&L breakdown by sector.}
\label{fig:CallOTMSectorSelect}
\end{figure}

We obtain the updated P\&L graphs in Figure~\ref{fig:Call OTM PnL Graphs Modified}. A few points in the right tails of the P\&L distributions are missing. Similarly, the upward jumps in the cumulative portfolio P\&L graphs are removed. In both the $5\%$ and $10\%$ Calls portfolios, the EWP now underperforms, followed by the the optimal portfolio with $\alpha=0$. The risk aversion values of $\alpha=5, 10$ have the best performance for $5\%$ Calls while the values $\alpha=2, 5$ outperform for $10\%$ Calls. Thus, in most cases, our diversification approach seems to achieve its purpose, mainly by shrinking the number of occurrences of all options expiring worthless at the same time and therefore reducing the mass at $-100\%$ of the P\&L distribution. However, when a large uptrend happens in a sector, the EWP outperforms our diversified portfolios because it catches the astronomical P\&L of individual stocks.

\begin{figure}[h!]
\centering
\includegraphics[scale=0.22]{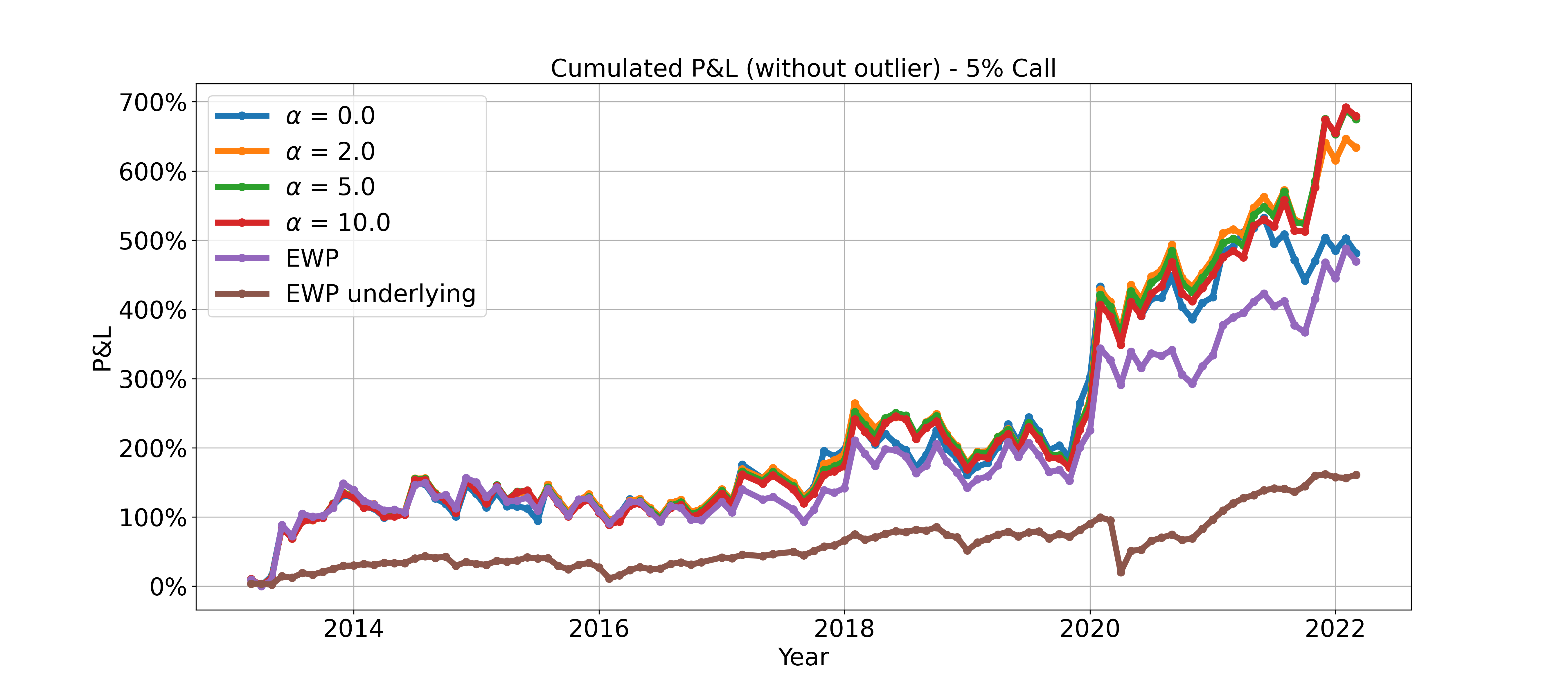}
\includegraphics[scale=0.22]{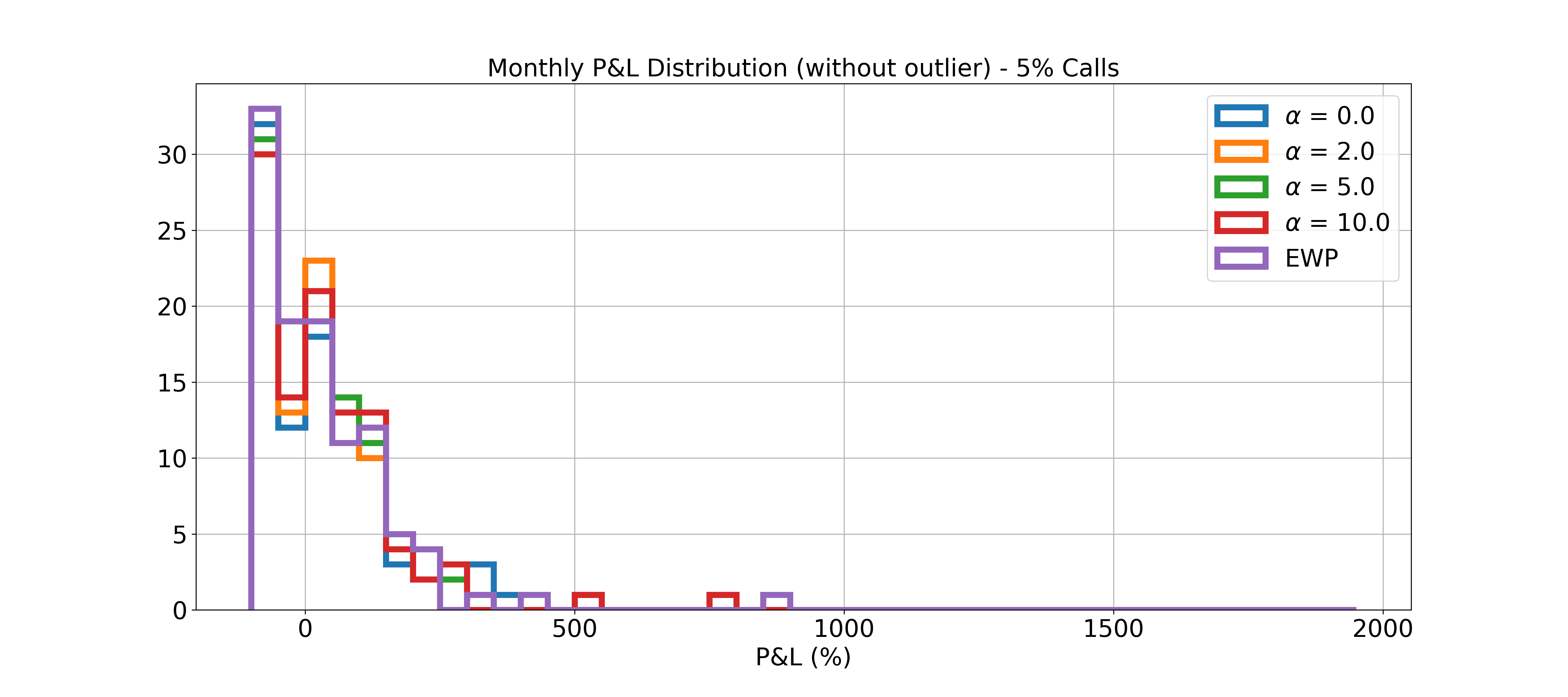}
\centering
\includegraphics[scale=0.22]{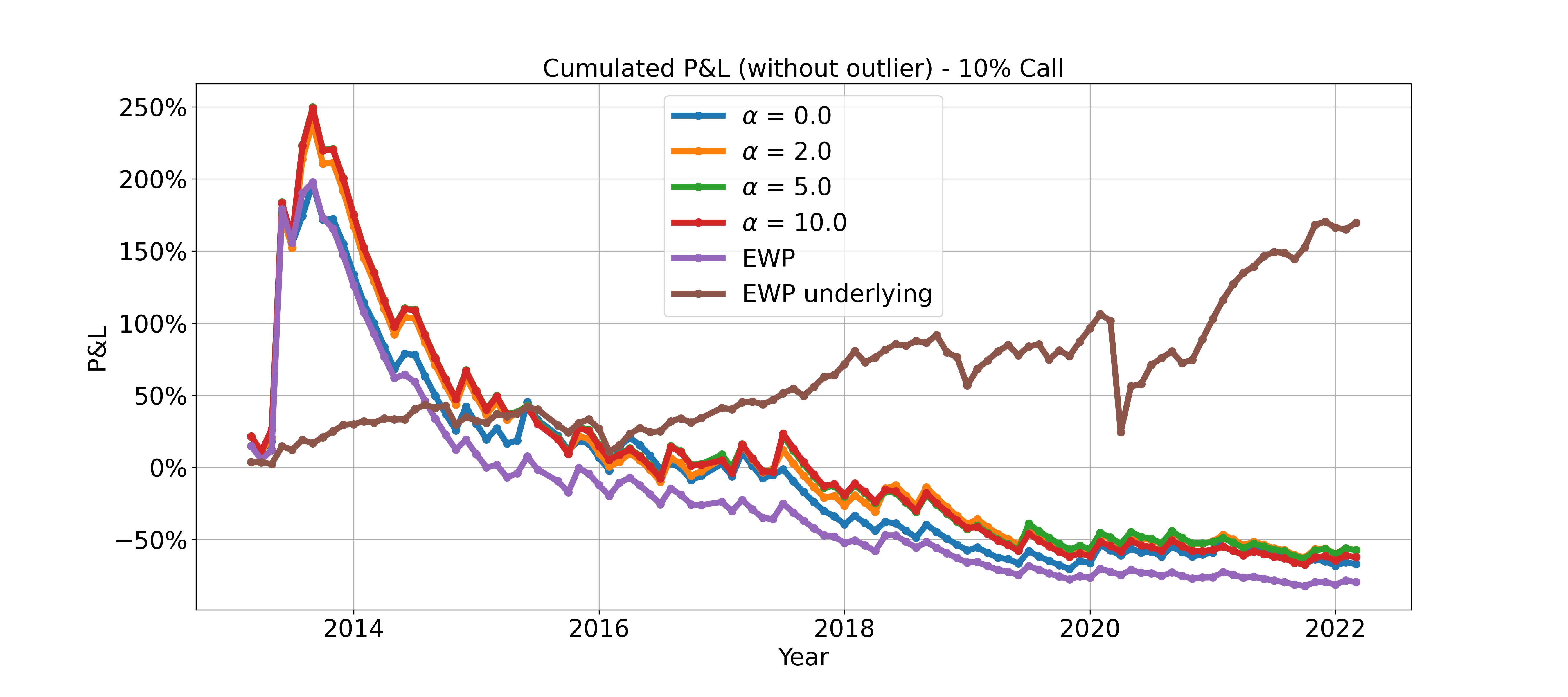}
\includegraphics[scale=0.22]{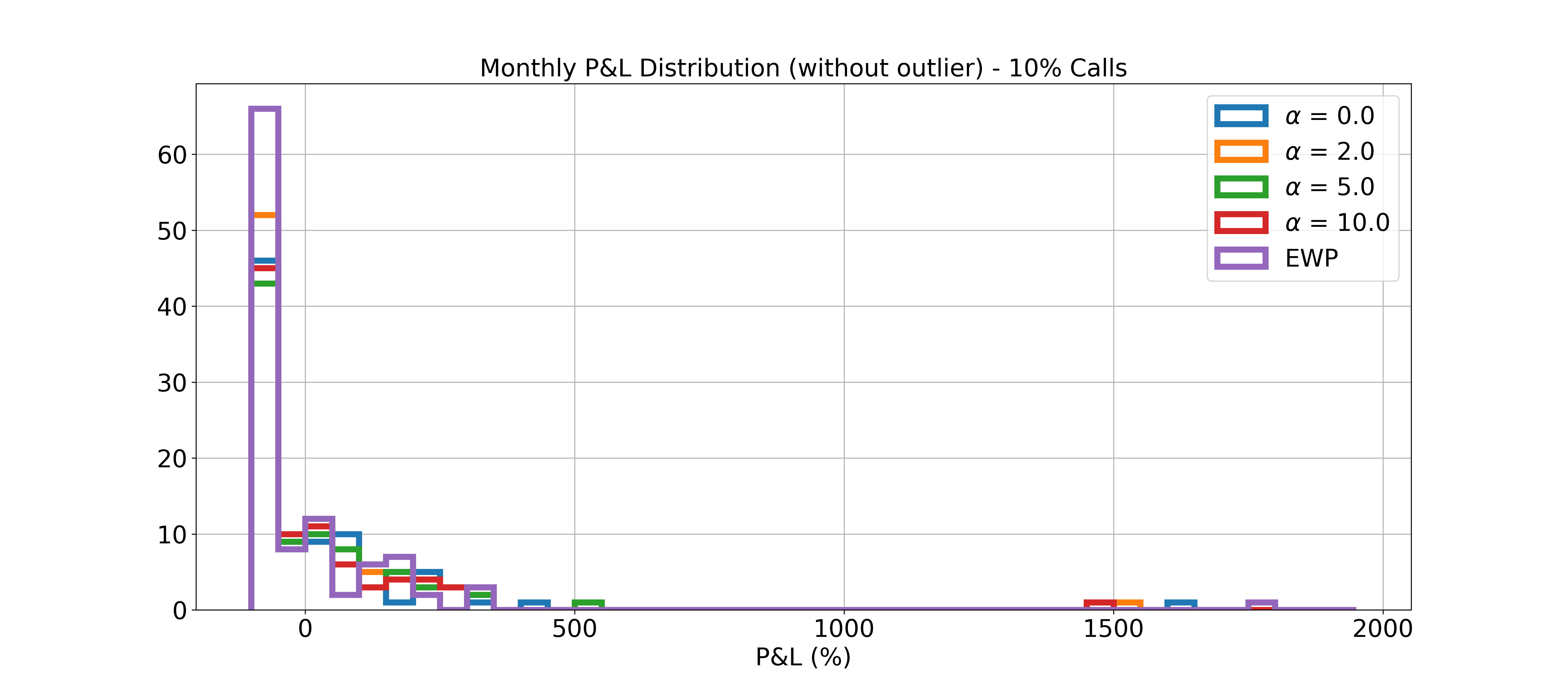}
\caption{Monthly cumulative P\&L and P\&L distribution without outliers; $5\%$ (top) and $10\%$ (bottom) Calls portfolio for $\alpha \in \{0, 2, 5, 10\}$, EWP Calls portfolio and EWP Underlying Equity portfolio.}
\label{fig:Call OTM PnL Graphs Modified}
\end{figure}

%%%%%%%%%%%%%%%%%%%%%%%%%%%%%%%%%%%%%%%%%%%%%%%%%%%%%%%%
\subsection{Mixed option portfolios}

We now test the optimisation on the same stocks using three different options; 
a $5\%$ OTM Call, $5\%$ OTM Put and a $10\%$ either-side Strangle, all bought with a month to maturity. Again, we compare the mixed options portfolio performance for different values of $\alpha \in \{0, 2, 5, 10\}$, as well as the EWP scheme and the EWPu. Figure~\ref{fig:Mixed PnL hist Raw} shows the P\&L distribution of our portfolios. We observe the huge mass at $-100\%$ for the EWP being shifted to the right for the optimal portfolios, which illustrates the diversification of options' payouts. Figure~\ref{fig:Mixed PnL cumul Raw} draws the cumulative portfolios P\&L. 
In February 2020, the optimal portfolio with $\alpha=2$ barely breaks even while the others have a negative P\&L. The following Covid crash brings significant returns as most Puts and deep OTM Strangles pay out. 
Throughout the entire backtest horizon, the portfolio with $\alpha=2$ seems to perform best, followed by $\alpha=10$, $\alpha=0$ and $\alpha=5$. 
EWP has the worst performance. In Table~\ref{tab:Stats Mixed Raw}, we observe extremely high values of kurtosis and skewness, even though the EWP has values above the other optimised portfolios. Again, the Sharpe Ratios are similar and the $\alpha=2$ portfolio has the highest CCRA value while the EWP has the lowest. 
Despite the usual downwards trends, the average monthly returns are positive because of the Covid expiration.

\begin{figure}[h!]
\centering
\includegraphics[scale=0.22]{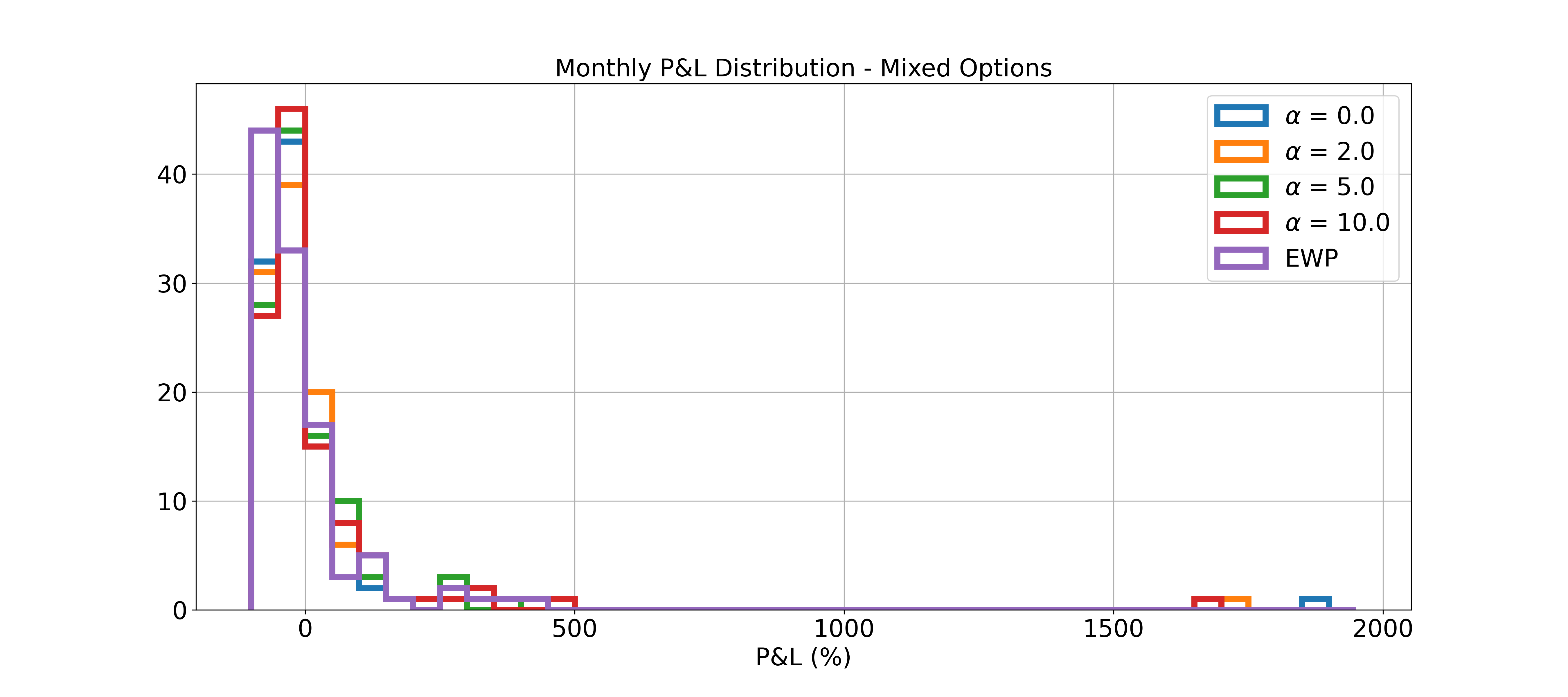}
\caption{Monthly P\&L distribution; Mixed options portfolio for $\alpha \in \{0, 2, 5, 10\}$, EWP Calls portfolio and EWP Underlying Equity portfolio.}
\label{fig:Mixed PnL hist Raw}
\end{figure}
\begin{figure}[h!]
\centering
\includegraphics[scale=0.22]{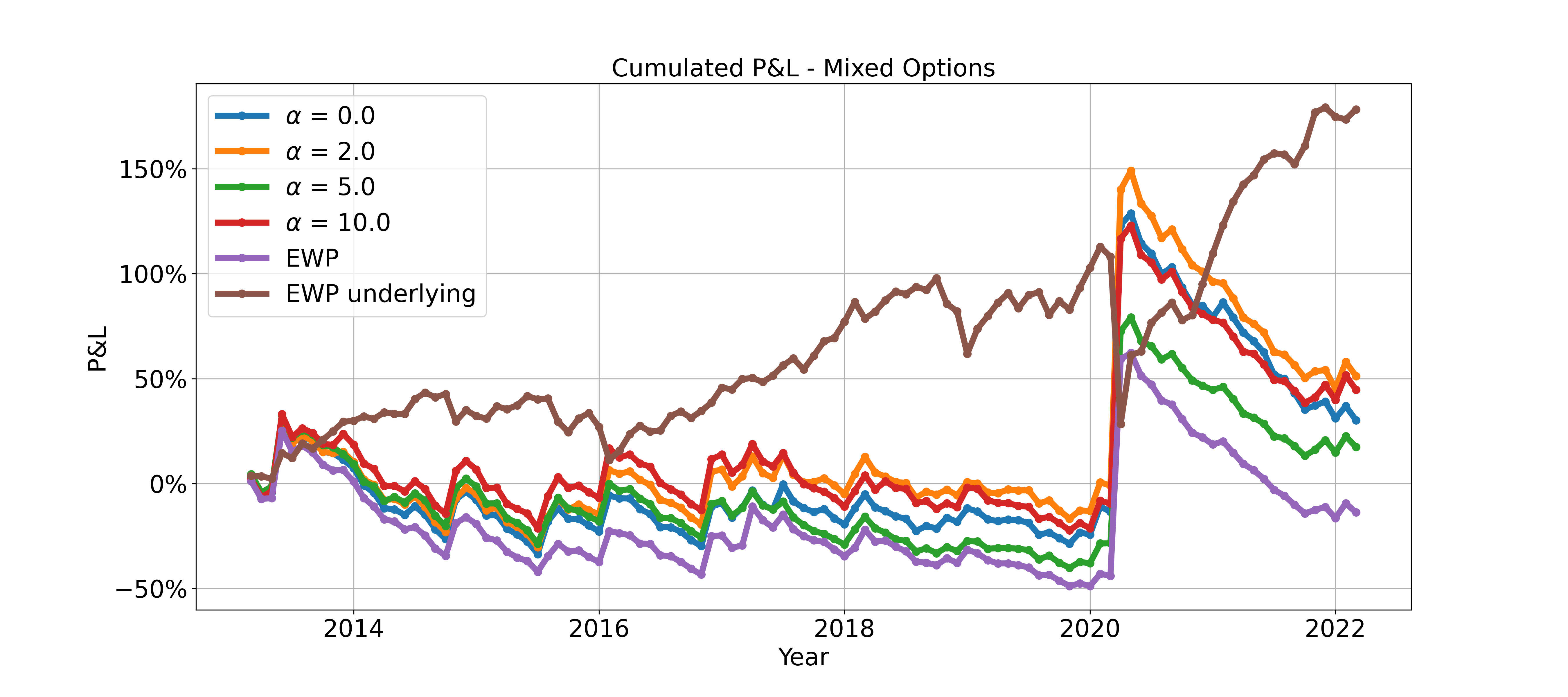}
\caption{Monthly cumulative P\&L; Mixed options portfolio for $\alpha \in \{0, 2, 5, 10\}$, EWP Calls portfolio and EWP Underlying Equity portfolio.}
\label{fig:Mixed PnL cumul Raw}
\end{figure}

Similarly, we analyse the P\&L breakdown by sector in Figure~\ref{fig:Mixed Sector}. 
The selected dates are 2016-11, 2017-03, 2017-06 and 2020-03. We observe the same pattern as in the case of~$5\%$ and~$10\%$ Calls where the P\&L contributions come from unique sectors. 
We also look at the Covid crash expiration where all sectors have unusually high returns and uneven returns, especially the Energy and Financial sectors. 
As such, we decide to discard these special data points.

\begin{figure}[h!]
\centering
\includegraphics[scale=0.22]{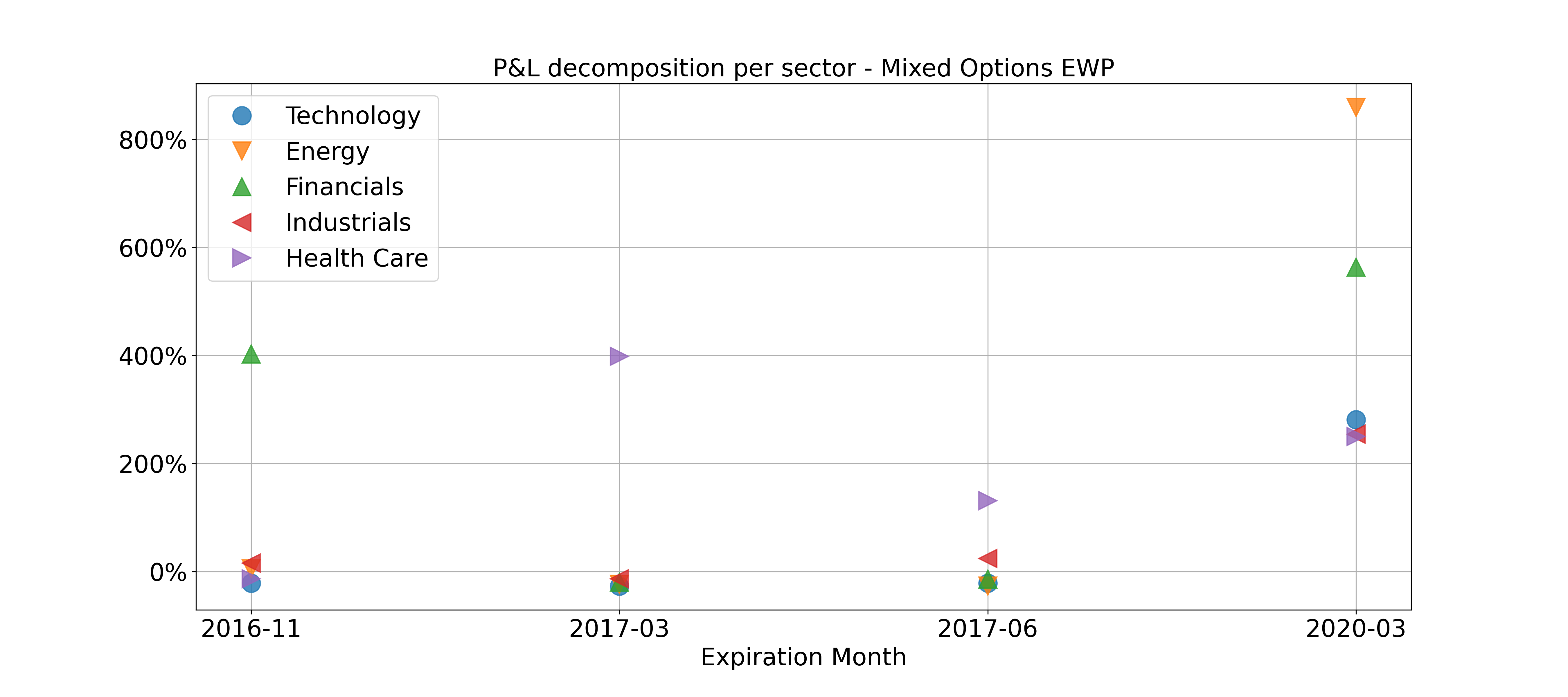}
\caption{Mixed Options EWP's P\&L breakdown by sector.}
\label{fig:Mixed Sector}
\end{figure}

The updated cumulative P\&L graph is plotted in Figure~\ref{fig:Mixed Cumul Pnl Modified}. 
Overall, after removing four expirations between 2016 and 2022, the P\&L follows a steady downward trend. 
Indeed, on the one hand, the total premium of four options is quite costly. 
On the other hand, only one side can pay out and in most cases, the deep $10\%$ OTM option will still expire worthless. 
In the end, the payout of a $5\%$ option is not sufficient to cover the cost of the strategy. We observe that the optimal portfolio with $\alpha=10$ still shows the best performance while the EWP has the worst P\&L. The optimal portfolio with $\alpha=0$ shows the second lowest P\&L and is outperformed by portfolios with a risk aversion penalty. In the updated Table~\ref{tab:Stats Mixed Modified}, we observe that the kurtosis and skewness values are much lower and the average monthly returns are negative.

\begin{figure}[h!]
\centering
\includegraphics[scale=0.2]{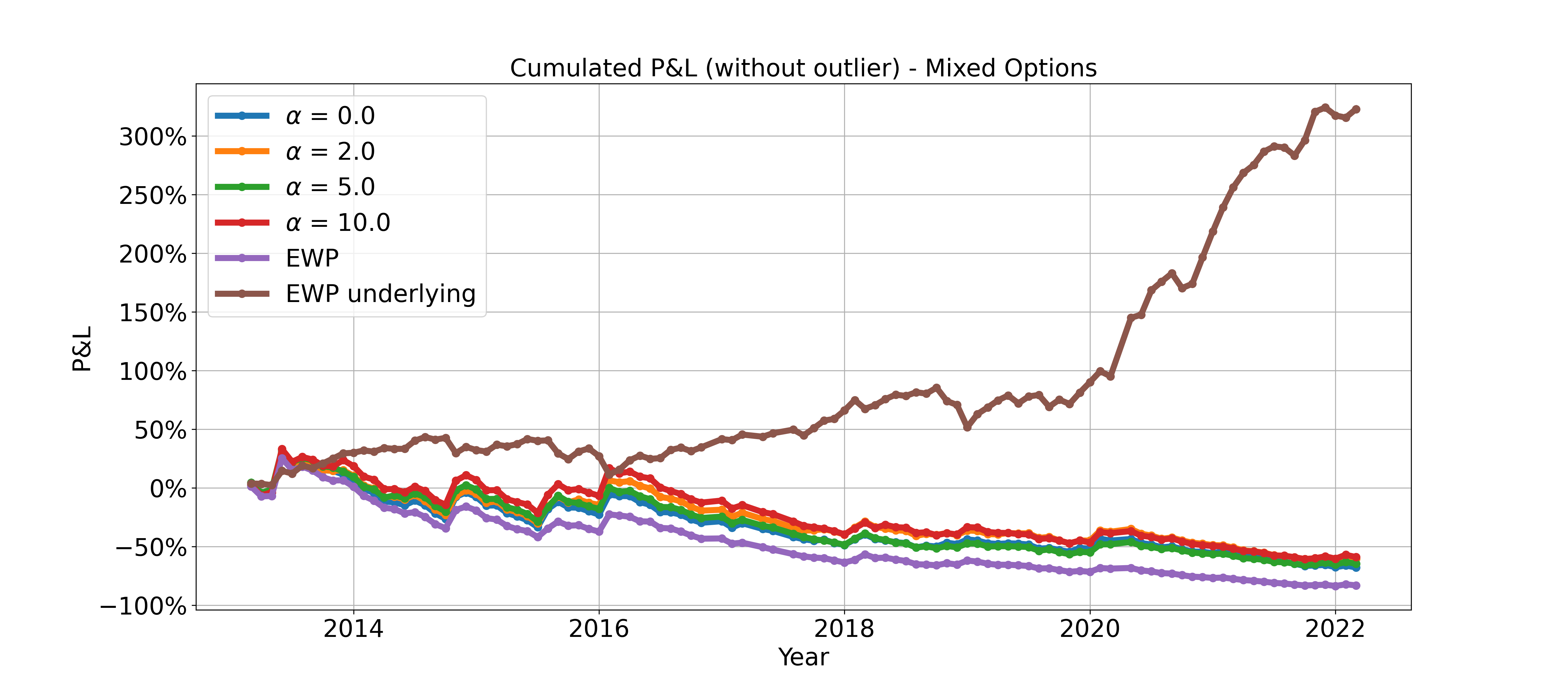}
\includegraphics[scale=0.2]{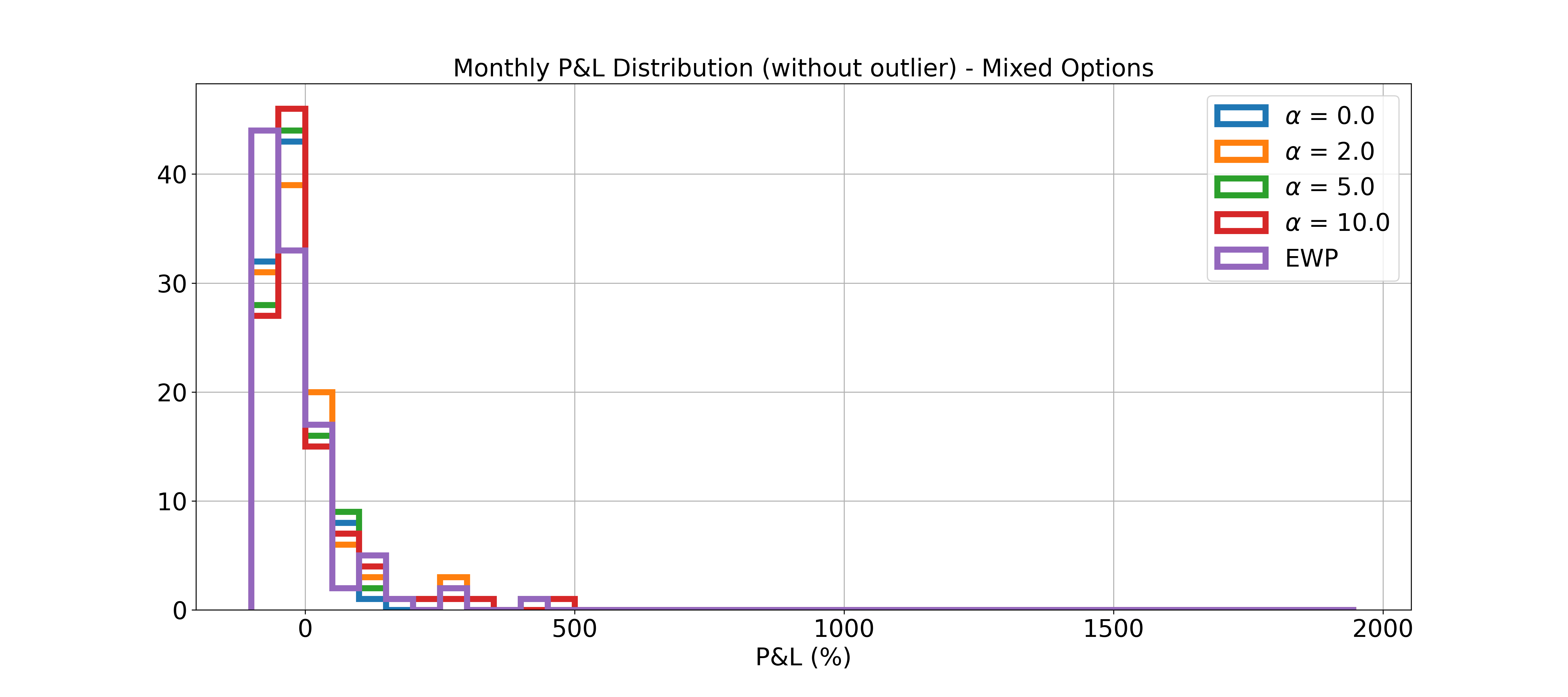}
\caption{Monthly cumulative P\&L (top) and P\&L distribution without outliers (bottom) for  the Mixed Options portfolio, the EWP portfolio and the EWP Underlying Equity portfolio.}
\label{fig:Mixed Cumul Pnl Modified}
\end{figure}

\begin{table}[h!]
\begin{tabular}{|c|c|c|c|c|c|} 
\hline
Statistics  & $\alpha = 0$ &$\alpha = 2$&$\alpha = 5$&$\alpha = 10$& EWP\\
 \hline
Total return ($\%$)            & 30.2  & 51.2  & 17.42  & 44.72  & -13.7  \\
 Average monthly return ($\%$) &  13.1 &  13.6 & 10.4  & 12.9  & 10.5  \\
 Skewness                      & 7.45  & 7.12  & 7.36  & 6.99  &  8.0 \\
 Kurtosis                      & 64.8  & 60.4  & 63.7  & 58.5  & 72.1  \\
 Annualised Sharpe Ratio                  &  0.22 & 0.25  & 0.20   &  0.24 &  0.16 \\
 CRRA (1e-4)                   &  -6.0 &  10.4 & -11.5  & 6.5  & -47.8    \\
 \hline
\end{tabular}
\caption{Mixed Options portfolio Statistics summary.}
\label{tab:Stats Mixed Raw}
\end{table}

\begin{table}[h!]
\begin{tabular}{|c|c|c|c|c|c|} 
\hline
Statistics  & $\alpha = 0$ &$\alpha = 2$&$\alpha = 5$&$\alpha = 10$& EWP\\
 \hline
Total return ($\%$)            &  -67.9 & -60.6  & -65.7  & -59.0  & -83.1  \\
 Average monthly return ($\%$) & -10.1  & -8.0  & -9.4  & -7.4  &  -17.7 \\
 Skewness                      &  2.59 & 2.41  & 2.43  & 2.65  & 2.74  \\
 Kurtosis                      & 8.51  & 7.72  & 8.0   & 9.36 & 9.69  \\
 Annualised Sharpe Ratio                  & -0.41  & -0.34  & -0.40  & -0.30  & -0.76  \\
 CRRA (1e-2)                   & -1.19  & -0.99  & -1.09  & -0.96  &    -1.80\\
 \hline
\end{tabular}
\caption{Mixed Options portfolio Statistics summary without outlier.}
\label{tab:Stats Mixed Modified}
\end{table}

%%%%%%%%%%%%%%%%%%%%%%%%%%%%%%%%%%%%%%%%%%%%%%%
\section*{Conclusion}
We provided a new methodology to optimise portfolios of European options, taking into account the specific behaviours of options compared to their underlyings.
Our key tool is a new dependency matrix, replacing the classical variance-covariance matrix, which allows to relate the
different option payoffs' conditional probabilities.
We provided closed-form expressions for such probabilities,
assuming a copula structure for the underlying securities.
We finally tested our method our historical data and showed that it is efficient and can be extended easily to large portfolios of (mixed) options.

%%%%%%%%%%%%%%%%%%%%%%%%%%%%%%%%%%%%%%%%%%%%%%
%%%%%%%%%%%%%%%%%%%%%%%%%%%%%%%%%%%%%%%%%%%%%%

\appendix

%%%%%%%%%%%%%%%%%%%%%%%%%%%%%%%%%%%%%%%%%%%%%%%%
\section{Additional figures}

%%%%%%%%%%%%%%%%%%%%%%%%%%%%%%%%%%%%%%%%%%%
\subsection{Dependency matrices}

We provide further empirical examples of dependency matrices~$\LLa$ for different options on S\&P 500 stocks, using data from the last eight years.
In Figure~\ref{fig:DepMatrixPut10}, 
we show a $10\%$ OTM Put option dependency matrix. 
The coefficients are extremely high in parts, first as $10\%$ price swings are very rare and second because the sampling period
includes an overall bull market for the S\&P.
The contrast can be seen in Figure~\ref{fig:DepMatrixCall10}
with the corresponding Call options,
for which the coefficients are much smaller,  reflecting that large positive returns were more probable than large negative returns over the period. 
Figure~\ref{fig:DepMatrixStrangle10}
shows a $5\%$ either-side Strangle. 
The coefficients are lower than in 
Figures~\ref{fig:DepMatrixPut10} and~\ref{fig:DepMatrixCall10} 
as we are effectively considering the sum of a Put and Call option, so the probability of payout is higher than for just Calls or Puts. 
We also clearly see dependence blocks, 
most significantly around Financials and Automobiles.

\begin{figure}
\centering
\includegraphics[scale=0.3]{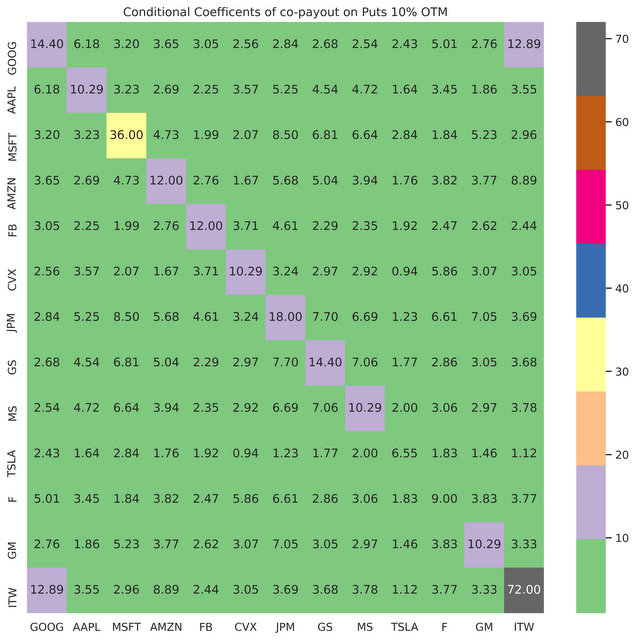}
\caption{Dependency matrix $10\%$ OTM Put options.}
\label{fig:DepMatrixPut10}
\end{figure}

\begin{figure}
\centering
\includegraphics[scale=0.3]{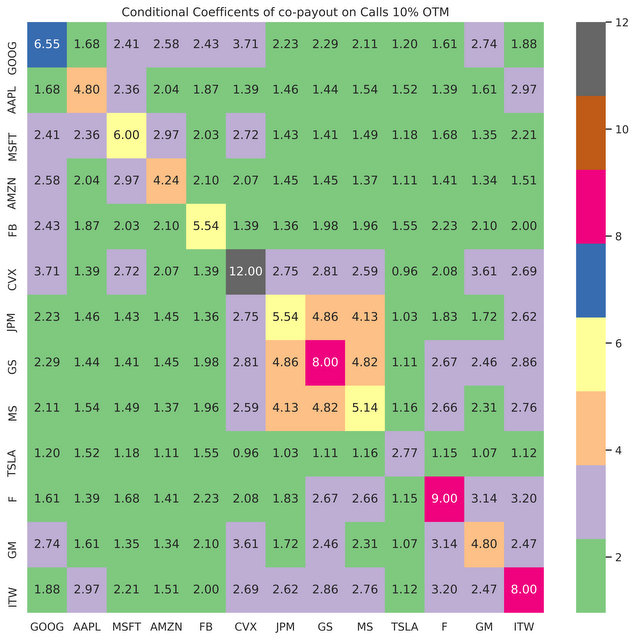}
\caption{Dependency matrix $10\%$ OTM Call options.}
\label{fig:DepMatrixCall10}
\end{figure}

\begin{figure}
\centering
\includegraphics[scale=0.3]{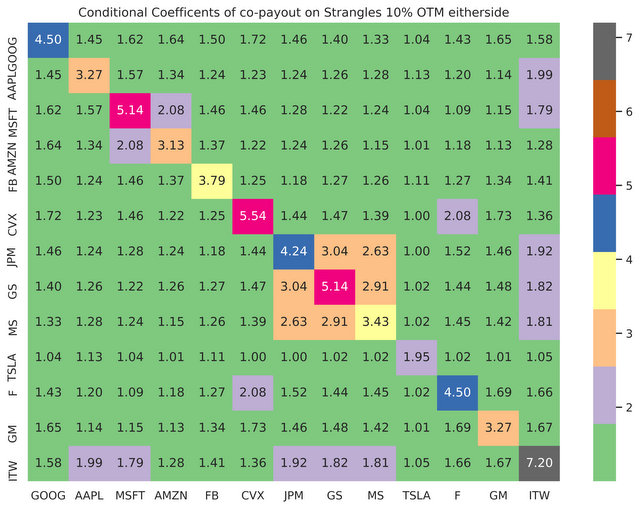}
\caption{Dependency matrix $10\%$ OTM either-side Strangle options.}
\label{fig:DepMatrixStrangle10}
\end{figure}

%%%%%%%%%%%%%%%%%%%%%%%%%%%%%%%%%%%%%
%%%%%%%%%%%%%%%%%%%%%%%%%%%%%%%%%%%%%
\newpage 
\bibliographystyle{siam}
\bibliography{Biblio}

\end{document}